\documentclass[11pt]{article}

\usepackage{tikz}

\usepackage{amsmath,amsfonts,amsthm,amssymb,color}

\newtheorem{theorem}{Theorem}[section]
\newtheorem{corollary}[theorem]{Corollary}
\newtheorem{lemma}[theorem]{Lemma}

\newtheorem{definition}[theorem]{Definition}

\theoremstyle{remark}
\newtheorem{remark}[theorem]{Remark}

\usepackage{latexsym,bbm,xspace,graphicx,float}
\definecolor{newblue}{rgb}{0.19, 0.55, 0.91}
\usepackage[colorlinks,citecolor=newblue,bookmarks=true,pagebackref=true, urlcolor=blue, linkcolor=blue, linktoc=page]{hyperref}
\usepackage[nameinlink]{cleveref}
\usepackage[letterpaper,margin=1in]{geometry}

\usepackage{url}
\usepackage{mathtools}
\usepackage{graphicx}
\usepackage{multirow}
\usepackage[ruled,linesnumbered]{algorithm2e}

\newcommand{\eps}{\varepsilon}

\newcommand{\bS}{\ensuremath\mathbb{S}}
\newcommand{\B}{\ensuremath\mathbf{B}}
\newcommand{\ms}{\ensuremath\mathsf{m_{space}}}
\newcommand{\Abs}[1]{\left|#1\right|}
\newcommand{\simplex}[1]{\ensuremath\Delta_{#1}}

\usepackage{amsmath,amsfonts,bm}

\def\1{\bm{1}}

\DeclareMathAlphabet{\mathsfit}{\encodingdefault}{\sfdefault}{m}{sl}
\SetMathAlphabet{\mathsfit}{bold}{\encodingdefault}{\sfdefault}{bx}{n}

\DeclareMathOperator*{\E}{\mathbb{E}}

\newcommand{\R}{\mathbb{R}}

\newcommand{\Var}{\mathbf{Var}}

\renewcommand{\tilde}{\widetilde}

\newcommand{\norm}[1]{\left\|#1\right\|}

\newcommand{\normtwo}[1]{\norm{#1}_2}

\providecommand{\expect}[2]{\ensuremath{\ifthenelse{\equal{#1}{}}{\mathbb{E}}{\mathbb{E}_{#1}}\!\left[#2\right]}\xspace}
\providecommand{\prob}[2]{\ensuremath{\ifthenelse{\equal{#1}{}}{\Pr}{\Pr_{#1}}\!\left[#2\right]}\xspace}

\newcommand{\inner}[1]{\langle #1\rangle}

\usepackage{bm}

\newcommand{\abs}[1]{\left|{#1}\right|}

\DeclareMathOperator{\poly}{poly}
\DeclareMathOperator{\polylog}{polylog}

\DeclareMathOperator{\scap}{Cap}

\newcommand{\Z}{\mathbb{Z}}  

\newcommand{\OL}{\mathsf{OL}}
\newcommand{\rowspan}{\mathrm{rowspace}}
\newcommand{\colspan}{\mathrm{colspace}}
\newcommand{\cN}{\mathcal{N}}
\DeclareMathOperator{\interior}{int}

\usepackage[inline]{enumitem}
\usepackage{array}
\newcolumntype{L}[1]{>{\raggedright\let\newline\\\arraybackslash\hspace{0pt}}m{#1}}
\newcolumntype{C}[1]{>{\centering\let\newline\\\arraybackslash\hspace{0pt}}m{#1}}
\newcolumntype{R}[1]{>{\raggedleft\let\newline\\\arraybackslash\hspace{0pt}}m{#1}}

\begin{document}

\author{Yi Li\footnote{Supported in part by Singapore Ministry of Education Tier 1 grant RG75/21}\\ \small{Division of Mathematical Sciences}\\ \small{Nanyang Technological University}\\  \small{\texttt{yili@ntu.edu.sg}}\\
	   \and
	   Honghao Lin\footnote{Supported by National Science Foundation (NSF) Grant CCF-1815840} $\qquad\qquad$ David P. Woodruff\footnote{Supported in part by National Science Foundation (NSF) Grant CCF-1815840} \\ \small{Computer Science Department}\\ \small{Carnegie Mellon University} \\  \small{\texttt{\{honghaol,dwoodruf\}@andrew.cmu.edu}}
	   }
	   
\date{\vspace{-5ex}}	   

\title{The $\ell_p$-Subspace Sketch Problem in Small Dimensions with Applications to Support Vector Machines}

\maketitle

\begin{abstract}
In the $\ell_p$-subspace sketch problem, we are given an $n \times d$ matrix $A$ with $n > d$, and asked to build a small memory data structure $Q(A,\eps)$ so that, for any query vector $x \in \mathbb{R}^d$, we can output a number in $(1 \pm \eps) \|Ax\|_p^p$ given only $Q(A,\eps)$. This problem is known to require $\tilde{\Omega}(d \eps^{-2})$ bits of memory for $d = \Omega(\log(1/\eps))$. However, for $d = o(\log(1/\eps))$, no data structure lower bounds were known. Small constant values of $d$ are particularly important for estimating point queries for support vector machines (SVMs) in a stream (Andoni et al.~2020), where only tight bounds for $d = 1$ were known.

We resolve the memory required to solve the $\ell_p$-subspace sketch problem for any constant $d$ and integer $p$, showing that it is $\Omega\big (\eps^{-\frac{2(d-1)}{d+2p}} \big )$ bits and $\tilde{O}\big (\eps^{-\frac{2(d-1)}{d+2p}} \big )$ words, where the $\tilde{O}(\cdot)$ notation hides $\poly(\log(1/\eps))$ factors. This shows that one can {\it beat} the $\Omega(\eps^{-2})$ lower bound, which holds for $d = \Omega(\log(1/\eps))$, for any constant $d$. Further, we show how to implement the upper bound in a single pass stream, with an additional multiplicative $\poly(\log \log n)$ factor and an additive $\poly(\log n)$ cost in the memory. Our bounds extend to loss functions other than the $\ell_p$-norm, and notably they apply to point queries for SVMs with additive error, where we show an optimal bound of $\tilde{\Theta}\big(\eps^{-\frac{2d}{d+3}}\big)$ for every constant $d$. This is a near-quadratic improvement over the $\Omega \big (\eps^{-\frac{d+1}{d+3}} \big )$ lower bound of Andoni et al. Further, previous upper bounds for SVM point query were noticeably lacking: for $d = 1$ the bound was $\tilde{O}(\eps^{-1/2})$ and for $d = 2$ the bound was $\tilde{O}(\eps^{-4/5})$, but all existing techniques failed to give any upper bound better than $\tilde{O}(\eps^{-2})$ for any other value of $d$. Our techniques, which rely on a novel connection to low dimensional techniques from geometric functional analysis, completely close this gap. 
\end{abstract}

\section{Introduction} \label{sec:intro}
We consider the subspace sketch problem, which is the problem of designing a low memory data structure to compress a given $n \times d$ matrix $A$, so that later given only the compressed version of $A$, one can query norms of vectors of the form $Ax$ for $x \in \mathbb{R}^d$. Formally,

\begin{definition}
In the subspace sketch problem, we are given an $n \times d$ matrix $A$ with entries specified by $O(\log (nd))$ bits, an accuracy parameter $\eps > 0$, and a function $\Phi: \mathbb{R}^n \rightarrow \mathbb{R}^{\geq 0}$, and the goal is to design a data structure  $Q_{\Phi}$ so that, with constant probability, simultaneously for all $x \in \mathbb{R}^d$, $Q_{\Phi}(x) = (1 \pm \eps) \Phi(Ax)$. 
\end{definition}

An important case of the above is when the functions correspond to the classical $\ell_p$-norms, i.e.,  $\Phi(x) = \sum_{i=1}^n |x_i|^p$ for some $p \geq 1$. 
A space bound of $d^{O(p)}$ words is known for even integers $p$, independent of $\eps$. For $p$ that is not an even integer, 
it was shown in \cite{DBLP:journals/siamcomp/LiWW21} that there is an $\tilde{\Omega}(d \eps^{-2})$ lower bound on the memory required to solve the subspace sketch problem for $d = \Omega(\log(1/\eps))$. Here and throughout, the $\tilde{O}(\cdot)$ notation hides poly-logarithmic factors in its arguments. The fact that $d = \Omega(\log(1/\eps))$ was crucial for the arguments in \cite{DBLP:journals/siamcomp/LiWW21}, and a natural question is if the same $\tilde{\Omega}(d \eps^{-2})$ lower bound holds for smaller $d$, in particular for constant $d$. 

The interest in constant $d$ is particularly motivated given the recent work of \cite{ABL+20}, which studied the support vector machines (SVM) problem in constant dimensions in the streaming setting. Here $x$ can be thought of as a pair $(\theta, b) \in \mathbb{R}^d \times \mathbb{R}$ and each of the $n$ rows of $A$ can be thought of as a pair $(x_i, y_i) \in \mathbb{R}^d \times \{-1,1\}$, and for a parameter $\lambda > 0$ we have:
\begin{eqnarray}\label{eqn:svm}
\Phi((\theta, b)) = \frac{\lambda}{2} \|(\theta, b)\|_2^2 + \frac{1}{n} \sum_{i=1}^n \max\{0, 1-y_i (\theta^T x_i + b)\}.
\end{eqnarray}

The authors refer to the above as the ``point query" version of SVM and show an $\Omega \big (\eps^{-\frac{d+1}{d+3}} \big )$ bit lower bound for any single-pass streaming algorithm for solving this problem. In fact, their lower bound applies to the memory required of any data structure for solving the subspace sketch problem with $\Phi$ as in (\ref{eqn:svm}). In terms of upper bounds, \cite{ABL+20} show an $\tilde{O}(\eps^{-1/2})$ bound for $d = 1$ and an $\tilde{O}(\eps^{-4/5})$ bound for $d = 2$. For any $d > 2$, the best known upper bound is a trivial $\tilde{O}(\eps^{-2})$ bound obtained by uniform sampling. In fact, these upper bounds are also all one-pass streaming algorithms. One of the major open questions of \cite{ABL+20} was to close this nearly quadratic gap for large constant $d$. 

The $\ell_p$-subspace sketch problem has also been studied in functional analysis for constant values of $d$ for the special case of $p = 1$, and for the special case of requiring an embedding, i.e., a low dimension $m$ and a matrix $B$ of $m$ rows so that $\|Bx\|_1 = (1 \pm \eps)\|Ax\|_1$ for all $x$. In particular, a dimension of $\tilde{O}\big (\eps^{-\frac{2(d-1)}{d+2}} \big )$ was established in a sequence of work
\cite{beckt83,gordon85,schechtman87,linhart89,BL88,Mat96}, while a matching lower bound for this particular type of subspace sketch (and for $p = 1$ and constant $d$) was shown in \cite{bourgain89}. There are many natural questions left open by the functional analysis work: (1) can the upper bound be made a streaming upper bound with a small amount of memory? (2) does the lower bound hold for arbitrary data structures, (3) can the arguments extend to $p > 1$, etc.?

Throughout the remainder of this section, we assume that $d\geq 2$ and $p\geq 1$ are constants.

\begin{table}[t]
    \centering
    \begin{tabular}{|l|l l| l l|}
    \hline
    \textbf{Parameters}& \multicolumn{2}{c|}{\textbf{Lower Bound}}  & \multicolumn{2}{c|}{\textbf{Upper Bound}}\\
    \hline 
         $ p = 1$ & \multirow{5}{*}{$\Omega\big(\eps^{-\frac{2(d - 1)}{d + 2p}}\big)$ } & \multirow{5}{*}{Theorem~\ref{thm:lower_bound}} &$\tilde{O}\big(\eps^{-\frac{2(d - 1)}{d + 2}}\big)$ & \cite{Mat96}\\
         $p \in \Z\setminus 2\Z$ & & &$\tilde{O}\big(\eps^{-\frac{2(d - 1)}{d + 2p}}\big)$ & Theorem~\ref{thm:coreset_multiplicative} \\
         $p\in [1, \infty) \setminus \Z$ & & & $\tilde{O}\big(\eps^{-\frac{2(d^q - 1)}{d^q + 2}}\big)$ & $q=\lceil \frac{p}{2}\rceil$, Section~\ref{sec:general_tensor_UB}\\
         $p \in [1, \infty) \setminus \Z$, $d \ge 5$ & & &$\tilde{O}\big(\eps^{-\frac{2(d - 1)}{d + 2}}\big)$ & Theorem~\ref{thm:d>=5} \\
         $p \in (d-1, \infty) \setminus \Z$ & & & $\tilde{O}\big(\eps^{-\frac{2d}{2p - d + 2}}\big)$ & Section~\ref{sec:p>d-1} \\
         \hline 
         $p \in 2\Z$ & \multicolumn{2}{c|}{ no dependence on $\eps$ } & $O(1)$ & \cite{schechtman:even_p,DBLP:journals/siamcomp/LiWW21} \\
         \hline
         SVM & $\Omega\big(\eps^{-\frac{2d}{d + 3}}\big)$&  Theorem~\ref{thm:lower_bound_svm} & $\tilde{O}\big(\eps^{-\frac{2d}{d + 3}}\big)$ & Theorem~\ref{thm:streaming_svm}\\
         \hline
    \end{tabular}
    \caption{A summary of existing results and our results for the $\ell_p$-subspace sketch problem. The lower bounds are in terms of the number of bits and the upper bounds are in the number of words.}
    \label{tab:summary_of_results}
\end{table}

\subsection{Our Results}
A summary of our results is provided in Table~\ref{tab:summary_of_results}.

In this paper, we show that the $\Omega\big (\eps^{-\frac{2(d-1)}{d+2}} \big )$ lower bound actually holds for any type of data structure for $p = 1$. Furthermore, for every $p \in [1,\infty)\setminus 2\Z$, we obtain a lower bound of  $\Omega\big (\eps^{-\frac{2(d-1)}{d+2p}} \big )$. 

\begin{theorem}
\label{thm:lower_bound_informal}
Suppose that $p\in [1,\infty)\setminus 2\Z$. Any data structure that solves the $\ell_p$-subspace sketch problem for dimension $d$ and accuracy parameter $\eps$ requires $\Omega(\eps^{-\frac{2(d - 1)}{d + 2p}})$ bits of space.
\end{theorem}

For every integer $p$, we obtain an $\tilde{O}\big (\eps^{-\frac{2(d-1)}{d+2p}} \big )$ upper bound, matching the lower bound up to logarithmic factors.

\begin{theorem}[Informal]
\label{thm:upper_bound_informal}
Suppose that $A \in \R^{n \times d}$ and $p$ is a positive integer. There is a polynomial-time algorithm that maintains a data structure using  $\tilde{O}(\eps^{-\frac{2(d - 1)}{d + 2p}})$ words of space, which solves the $\ell_p$-subspace sketch problem. 
\end{theorem}

Moreover, we show the upper bound above can be implemented in a single pass row-arrival stream, with an additional multiplicative $\poly(\log \log n)$ factor and an additive $\poly(\log n)$ cost in the memory.

\begin{theorem}[Informal]
\label{thm:streaming_informal}
Let $A = a_1 \circ \cdots \circ a_n$ be a stream of $n$ rows, where $a_i \in \R^{1 \times d}$. There is an algorithm that maintains a data structure in $\tilde{O}(\eps^{-\frac{2(d - 1)}{d + 2p}})$ words of space, which solves the $\ell_p$-subspace sketch problem.
Moreover, the algorithm can be implemented in $\tilde{O}(\eps^{-\frac{2(d - 1)}{d + 2p}} + \log^{\frac{3d + 2p - 2}{d + 2p}} n)$ words of space. 
\end{theorem}

We also obtain an $O(1)$-update time algorithm with a slightly worse $\tilde{O}(\eps^{-\frac{2(d - 1)}{d + 2p - 1}})$ words of space bound for general $d = O(1)$ and a tight space bound of  $\tilde{O}(\eps^{-\frac{2(d - 1)}{d + 2p}})$ for $d \le 2p + 2$.

\begin{theorem}[Informal]
Let $A = a_1 \circ \cdots \circ a_n$ be a stream of $n$ rows, where $a_i \in \R^{1 \times d}$. There is an algorithm which maintains a data structure $Q$ of $\tilde{O}(\eps^{-\frac{2(d - 1)}{d + 2p - 1}})$ words of space which solves the $\ell_p$-subspace sketch problem.
Moreover, the algorithm updates $Q$ in $O(1)$ time and  can be implemented using  $\tilde{O}(\eps^{-\frac{2(d - 1)}{d + 2p - 1}})$ words of space.

When $d\leq 2p+2$, the size of $Q$ can be improved to $\tilde{O}(\eps^{-\frac{2(d - 1)}{d + 2p}})$ words of space and the whole algorithm can be implemented in $\tilde{O}(\eps^{-\frac{2(d - 1)}{d + 2p}})$ words of space.
\end{theorem}

To obtain a tight bound for the SVM point query problem mentioned above, we also study the following version of the affine $\ell_p$-subspace sketch problem, where 
\[
\Phi(x, b) = \sum_{i = 1}^n |b - \inner{A_i, x}| \;.
\]
for every $x \in \R^d$ and $b \in \R$. We show a tight space complexity of $\tilde{\Theta}(\eps^{-\frac{2d}{d + 2p + 1}})$.

\begin{theorem}[Informal]
\label{thm:affine_informal}
Suppose that $p\in [1,\infty)\setminus 2\Z$. Any data structure that solves the affine $\ell_p$-subspace sketch problem for dimension $d$ and accuracy parameter $\eps$ requires $\Omega(\eps^{-\frac{2d}{d + 2p + 1}})$ bits of space.

When $p$ is an integer, there is a polynomial-time algorithm that returns a data structure with $\tilde{O}(\eps^{-\frac{2d}{d + 2p + 1}})$ words of space, which solves the affine $\ell_p$ subspace sketch problem. 
\end{theorem}

Based on these results, we show that a tight space bound for the point query problem can be derived via a black box reduction with $p=1$.

\begin{theorem}[Informal]
\label{thm:svm_informal}
Any data structure which solves the $d$-dimensional point estimation problem for SVM requires $\Omega(\eps^{-\frac{2d}{d + 3}})$ bits of space.

Furthermore, there is an algorithm that maintains a data structure using $\tilde{O}(\eps^{-\frac{2d}{d + 3}})$ words of space, which solves the $d$-dimensional point query problem for SVM. The algorithm can be implemented using $\tilde{O}(\eps^{-\frac{2d}{d + 3}})$ words of space. 
\end{theorem}

Lastly, we obtain results for non-integer $p$, giving algorithms with $o(\eps^{-2})$ words of space even for such $p$. The details are postponed to Appendix~\ref{sec:non-integer}.

\subsection{Our Techniques}
One of our key technical contributions is to connect the SVM point query problem to the $\ell_1$-subspace sketch problem and to use techniques for the latter from geometric functional analysis, which previously had not been considered in the context of the SVM problem~\cite{ABL+20}.
Throughout this section, we assume that $p\geq 1$ is not an even integer.

\subsubsection{Lower Bound}
The idea behind our lower bound for the subspace sketch problem is to give Alice one of $m = 2^{\Omega(n)}$ possible subsets $S_1, \ldots, S_m$ of $n/2$ points of a fixed set $S = \{p_1, \ldots, p_n\}$ of $n$ points on the unit sphere $\bS^{d-1}$,
where $\|p_i - p_j \|_2 \geq \eta$ for all $i \neq j$. Here we have $|S_i \cap S_j| \leq n/4$ for every pair $i \neq j$. If Alice has a specific subset $S_i$, she can send the subspace sketch of her set to Bob. Bob then pretends he has an $S_j$ and enumerates over all possible queries $x$, and by construction of our sets $S_i$, we will (abusing notation and writing a set as the matrix whose rows correspond to the entries in the set) have that $\abs{\|S_i x\|_p - \|S_j x\|_p}$ is larger than the tolerable subspace sketch error, and Bob will be able to determine that $i \neq j$. 

Our main novelty is Lemma \ref{lem:min_error}, which says that for two sets $A$ and $B$ of points on the sphere, each symmetric around the origin and such that no point in $A$ is close to any point in $B$, there is some direction $x$ on the sphere for which $|\|Ax\|_p^p - \|Bx\|_p^p|$ is large. 
The proof of this lemma is inspired by ideas from geometric functional analysis \cite{bourgain89}, which give lower bounds for the specific case when the data structure is an embedding for $p = 1$ and for constant $d$. Indeed, as in their proof, we make use of spherical harmonics. However, we require a significant strengthening of the arguments in \cite{bourgain89}. In particular, the lemma in \cite{bourgain89} can be seen as lower bounding $\max_x \|Ax\|_p^p - \min_x \|Ax\|_p^p$, which, after expanding $\norm{Ax}_p^p$ in a spherical harmonic series, boils down to lower bounding
\[
\sum_{i,j} f(\inner{A_i,A_j}),\quad \text{where} \quad f(t) = \frac{1-r^4}{(1 + r^4 - 2r^2 t)^{d/2}},
\]
where $r$ is a parameter to be determined. The lower bound of \cite{bourgain89} is obtained simply by considering only the terms for which ``$i=j$''.  In our case, we need to lower bound $\max_x (\norm{Ax}_p^p - \norm{Bx}_p^p) - \min_x (\norm{Ax}_p^p - \norm{Bx}_p^p)$, which reduces to lower bounding a more complicated quantity:
\begin{equation}\label{eqn:lb_intro_aux}
\sum_{i,j} f(\inner{A_i,A_j}) + \sum_{i,j} f(\inner{B_i,B_j}) - \sum_{i,j} f(\inner{A_i,B_j}).
\end{equation}
The first two terms can be lower bounded similarly by taking only the ``$i = j$" terms as before. However, there is a third term which causes additional complications since it requires a good upper bound. To do this, for each point $A_i$, we partition the points $B_j$ into level sets of geometrically increasing distances from $A_i$. The critical observation is that the number of points in each level set grows at a slower rate than the function $f$ decays. Hence, the contribution from each level set is geometrically decreasing and the total contribution for each fixed $A_i$ can thus be controlled. In the end, we are able to show that the third term in \eqref{eqn:lb_intro_aux} is at most a constant fraction of the first two terms, which leads to our desired lower bound.

\subsubsection{Upper Bound}
Our upper bound is inspired from an argument of Matousek~\cite{Mat96} for $p=1$. We give a high-level description of this idea, assuming first that each point $A_i\in \bS^{d-1}$. The points $A_i$ are partitioned into a number of groups $P_1,\dots,P_s$ each of diameter at most $\eta = \eps^{2/(d+2)}$ (Lemma~\ref{lem:partition}), then $\Theta(d)$ random points are chosen from each group, such that the barycenter of the randomly selected points is the same as the barycenter of the group (Lemma~\ref{lem:weights}). The data structure stores the selected points in each group as a surrogate for the group.

For a fixed query point $x$, each group $P_j$ belongs to one of three types, based on its relative position to the the equator $\{y: \inner{x,y} = 0\}$: positive type, if $\inner{A_i,x} > C\eta$ for all $A_i\in P_j$; negative type, if $\inner{A_i,x} < -C\eta$ for all $A_i\in P_j$; zero type, if $\abs{\inner{A_i,x}} \leq C\eta$ for all $A_i\in P_j$, where $C > 0$ is an absolute constant. When $P_j$ is of positive type, we have $\sum_{A_i\in P_j} \abs{\inner{A_i,x}} = \sum_{A_i\in P_j} \inner{A_i,x} = \inner{\sum_{A_i\in P_j} A_i,x}$, which can be calculated exactly without error from the sampled points, as guaranteed by the barycenter property. When $P_j$ is of negative type, the contribution can be calculated exactly in a similar manner. Next, consider the groups $P_j$ of zero type. Since all summands $\abs{\inner{A_i,x}}$ are small, a Bernstein bound shows that using randomly selected points can approximate the \emph{total} contribution from \emph{all} zero-type groups up to a small additive error with high probability. Taking a union bound over a net for query points $x$, the overall sum $\sum_i \abs{\inner{A_i, x}}$ can be estimated with a small additive error $\eps$ with high probability simultaneously for all $x\in \bS^{d-1}$.

Now, suppose that $p$ is an odd integer. 
The key observation is a ``tensor trick''
\[
\inner{x, y}^p = \inner{x^{\otimes p}, y^{\otimes p}},
\]
where $x^{\otimes p}$ and $y^{\otimes p}$ are $d^p$-dimensional vectors. Lemma~\ref{lem:weights} is then applied to a group of $d^p$-dimensional points, so $\Theta(d^p)$ points are selected randomly in each group and stored in the data structure. Another crucial observation is that in this way, all the error comes from the terms $|\inner{A_i, x}|^p$ such that $|\inner{A_i, x}| \le \eta$, which implies $|\inner{A_i, x}|^p \le \eta^p \ll \eta$. This allows for tighter concentration than what is possible for $p = 1$. 
We can therefore estimate $\sum_i \abs{\inner{A_i,x}}^p$ up to a small additive error $\eps$ for all $x\in\bS^{d-1}$, assuming that $\norm{A_i}_2 = O(1)$ for all $i$.

The procedure above can be generalized to estimate $\sum_i w_i\abs{\inner{A_i,x}}^p$ up to an additive error of $\eps\sum_i w_i$, where $w_i\geq 0$ is the weight associated with the point $A_i$. This requires that the random points selected from each group carry (new) weights such that the weighted barycenters are the same. This was already attained in Matousek's work (Lemma~\ref{lem:weights}).

In order to obtain a multiplicative error for a general matrix $A$, we perform a preprocessing step, which transforms the John ellipsoid of $\{x\in \R^d: \norm{Ax}_p \leq 1\}$ to the unit $\ell_2$ ball in $\R^d$ via a linear transformation, giving a matrix $A'$ for which we can show that $\norm{A'_i}_2 = O(1)$ and $\norm{A'x}_p^p= \Omega(\sum_i \norm{A_i'}_2^p)$. This is sufficient to deduce that the additive error to the normalized version of $A'$ is in fact a multiplicative error for $A$.

We remark that having a $p$-th power in $\abs{\inner{A_i,x}}^p$ is only useful when $\inner{A_i,x}$ is small, and improvements exploiting this benefit may not occur in other algorithms. 
For example, Matousek proposed another algorithm~\cite{Mat96} which removes the $\log(1/\eps)$ factors for $d \ge 5$, obtaining a clean $O(\eps^{-\frac{2(d-1)}{d+2}})$ bound. The analysis of this algorithm relies on the fact that the function $x\mapsto \abs{\inner{u,x}} - \abs{\inner{v,x}}$ is Lipschitz on the region where $\inner{u,x}$ and $\inner{v,x}$ have the same sign and the Lipschitz constant is proportional to $\norm{u-v}_2$. However, we cannot expect a substantially smaller Lipschitz constant for the function $x\mapsto \abs{\inner{u,x}}^p-\abs{\inner{v,x}}^p$. Interestingly, as we shall show in Appendix~\ref{sec:non-integer}, despite it  failing to give a tight bound for integer $p > 1$, this algorithm actually implies an $O\big(\eps^{-\frac{2(d - 1)}{d + 2}}\big)$ upper bound for all constant $p > 1$ once $d \ge 5$.

\subsubsection{Streaming Upper Bound} 

The preceding approach for the upper bound does not give a streaming algorithm since computing the groups $P_1,\dots,P_s$ to perform the partition requires storing all of the points. Nevertheless, it can be viewed as a coreset procedure which, given a set of weighted points $(A,w)$ (where $w$ is a vector in which $w_i$ is the weight for $A_i$), outputs a small subset $B\subseteq A$ of size $\tilde{O}(\eps^{-\frac{2(d-1)}{d+2p}})$ together with new weights $w'$ such that $\sum_j w_j' \abs{\inner{B_j,x}}^p  = (1\pm\eps)\sum_i w_i \abs{\inner{A_i,x}}^p$. The standard merge-and-reduce framework (see, e.g., \cite{BDM+20} in the context of numerical linear algebra) can then be applied, leading to a streaming algorithm using $\tilde{O}(\eps^{-\frac{2(d-1)}{d+2p}}\polylog(n/\eps))$ words of space. However, this memory bound depends on the product of a term involving $\eps$ and a term involving $\log n$. By instead running the algorithm with $\eps = \Theta(1)$ and using it to estimate the so-called $\ell_p$-sensitivities, according to which we can sample the points $A_i$, we can replace $n$ with $\poly(d/\eps) \log n$, and then run the merge-and-reduce framework on this new value of $n$, resulting in only a $\poly(\log \log n)$ factor multiplying the term depending on $\eps$, plus an additive $\poly(\log n)$ term. 

\subsubsection{Connection to SVM}
As mentioned, one of our key contributions is showing that the SVM point query problem can be related the $\ell_1$-subspace sketch problem via a black-box reduction. As shown in~\cite{ABL+20}, the function $\Phi$ for SVM can be modified to
\[
\Phi(\theta, b) = \frac{1}{n} \sum_{i=1}^n \max\{0, b - \theta^T x_i \} \;,
\]
without loss of generality.
Consider the special case when $b = 0$. Suppose that $X = \{x_i\}$ is the point set given by the data stream. Let $-X = \{-x: x\in X\}$ and observe that
\[
\underset{X}{\Phi(\theta, 0)} + \underset{-X}{\Phi(\theta, 0)} = \frac{1}{n}\sum_i \left(\max\{0, \theta^\top x_i\} + \max\{0, -\theta^\top x_i\}\right) = \frac{1}{n} \sum_i \left|\theta^\top x_i\right| \;,
\]
which means that a lower bound for the $d$-dimensional $\ell_1$-subspace actually yields a lower bound for the $d$-dimensional point query problem for SVM. To obtain a tight lower bound, we instead study a special affine version of the $\ell_p$-subspace sketch problem, where 
\[
\Phi(x, b) = \sum_{i = 1}^{n} |b - \inner{A_i, x}|
\]
for a given $x \in \R^d$ and $b \in \R$.

The key observation is the following. Given a matrix $A \in \R^{n \times d}$, let $B \in \R^{n \times (d + 1)}$ be the matrix in which the $i$-th row $B_i = \begin{pmatrix}A_i & -1\end{pmatrix}$. Suppose that $x \in \R^d$ and $b \in \R$ are the query vector and value, and let $y  = \begin{pmatrix}x^\top & b\end{pmatrix}^\top \in \R^{d + 1} $. Then
\[
\sum_i |\inner{A_i, x} - b|^p = \sum_i |\inner{B_i, y}|^p = \norm{By}_p^p.
\]
Hence, the $d$-dimensional affine $\ell_p$-subspace sketch is related to the $(d + 1)$-dimensional $\ell_p$-subspace sketch problem where for all points $A_i$, the last coordinate is $1$.
A closer examination of the lower bound for the $\ell_p$ subspace sketch problem reveals that the lower bound for $\max_x (\norm{Ax}_p^p - \norm{Bx}_p^p) - \min_x (\norm{Ax}_p^p - \norm{Bx}_p^p)$ continues to hold, up to a constant factor, even when $A$ and $B$ do not necessarily lie on $\bS^{d-1}$ but rather in a thin spherical shell, i.e., $\norm{A_i}_2, \norm{B_i}_2 = \Theta(1)$, provided that their projections on $\bS^{d-1}$, $A_i/\norm{A_i}_2$ and $B_i/\norm{B_i}_2$, are sufficiently separated from each other. Hence, the point set $S$ in the communication problem can now be chosen from the spherical cap $\{x\in\bS^{d}: x_{d+1} = \Omega(1)\}$ so that normalizing the last coordinate $x_{d+1}$ to $1$ yields a vector $x' = x/x_{d+1}$ lying inside the spherical shell. The same lower bound (up to a constant factor) for the $(d+1)$-dimensional $\ell_p$ sketch problem then follows. 

\section{Preliminaries} \label{sec:prelim}

\paragraph{Notation} Let $\bS^{d-1}$ denote the unit sphere in $\R^d$, i.e., $\bS^{d-1} = \{x\in \R^d: \norm{x}_2 = 1\}$, and let  $\simplex{d-1}$ denote the standard $(d-1)$-simplex, i.e., $\simplex{d-1} = \{x\in \R^d: x_1+\cdots+x_d = 1\text{ and }x_i\geq 0\text{ for all }i\}$.

For two matrices $A$ and $B$ of the same number of columns, we denote their vertical concatenation by $A\circ B$.

\paragraph{Spherical Harmonics} The spherical harmonics $\{Y_{k,j}\}_{k,j}$ form an orthonormal basis in $L^2(\bS^{d-1},\sigma_{d-1})$, where $\sigma_{d-1}(x)$ denotes the normalized rotationally-invariant measure on $\bS^{d - 1}$. Here $k\geq 0$ is an integer and $j = 1,\dots, M(d,k)$ for each $k$, where $M(d,k) = \binom{k+d-2}{d-2} + \binom{k+d-3}{d-2} = O(k^{d-2})$. The following are some useful properties of spherical harmonics (see, e.g.~\cite{atkinson12,DX13}).
\begin{itemize}
    \item (Addition Theorem) For all $x,y\in \bS^{d-1}$,
\begin{equation}
    \label{eq:1}
    \sum_{j} Y_{k, j}(x) Y_{k,j}(y) = M(d, k) P_{k,d}(\inner{x,y}),
\end{equation}
where $P_{k,d}(t)$ is the Legendre polynomial of degree $k$ in dimension $d$.
    \item (Funk-Hecke Formula) Suppose that $f:[-1,1]\to\R$ is a function and $y\in \bS^{d-1}$. For $h:\bS^{d-1}\to \R$ defined as $h(x)=f(\langle x,y\rangle)$, it holds that
    \begin{equation}
    \label{eq:funk-hecke}
    \inner{h, Y_{k, j}} = \int_{\bS^{d - 1}} h(x) Y_{k, j}(x) \, d\sigma_{d   -1}(x) = \lambda_k Y_{k, j}(y),
    \end{equation}
    where 
    \begin{equation}\label{eq:lambda_k}
    \lambda_k = \frac{\Gamma(\frac{d}{2})}{\sqrt{\pi}\Gamma(\frac{d-1}{2})}\int_{-1}^1 f(t)(1-t^2)^{\frac{d-3}{2}} P_{k,d}(t) dt.
    \end{equation}
    \item (Poisson Identity) It holds for all $0\leq r<1$ and all $t\in [-1,1]$ that
    \begin{equation}
        \label{eq:3}
        \frac{1 - r^2}{(1 + r^2 - 2rt)^{d/2}} = \sum_{k = 0}^{\infty} M(d, k) r^k P_{k,d}(t).
    \end{equation}
\end{itemize}

\begin{lemma}\label{lem:lambda_k}
Suppose that $d\geq 2$ and $p\geq 1$ are constants. Let $f(t) = |t|^p$ and $\lambda_k$ be as defined in~\eqref{eq:lambda_k}. It holds that (i) $\lambda_k = 0$ for odd $k$; (ii) when $p$ is not an even integer, $\lambda_k\neq 0$ for all even $k$ and
\[
|\lambda_k| \sim_{d,p} \frac{1}{k^{d/2+p}} \sin \frac{\pi p}{2}, \quad \text{even }k\to\infty;
\]
(iii) when $p$ is an even integer, $\lambda_k = 0$ for all $k > p$.
\end{lemma}
\begin{proof}
Since $P_{k,d}(t)$ is an odd function when $k$ is odd, it is clear that $\lambda_k = 0$ when $k$ is odd. We shall assume that $k$ is even in the rest of the proof. In this case $P_{k,d}$ is an even function.

Note that $P_{n,k}$ is the normalized Gegenbauer polynomial, $P_{k,d}(t) = C_k^{\alpha}(t)/C_k^{\alpha}(1)$ with $\alpha = d/2-1$. It is known that $C_k^\alpha(1) = (2\alpha)_k/(k!)$. Invoking the identity (2.21.1.1) from \cite{prudnikov}, we have that
\[
\lambda_k = c(d) (-1)^{k/2}\Gamma\left(\alpha+\frac{1}{2}\right)\frac{\Gamma(\frac{p+1}{2})}{\Gamma(1+\frac{p}{2}+\alpha+\frac{k}{2})}\left(\frac{-p}{2}\right)_{k/2} .
\]
Now it is clear that $\lambda_k=0$ when $p$ is an even integer and $k > p$, and $\lambda_k\neq 0$ for all even $k$ when $p$ is not an even integer. In the latter case, when $k > 0$,
\[
\lambda_k = c(d, p) (-1)^{k/2+1} \sin\left(\frac{\pi p}{2}\right)\frac{\Gamma(\frac{k-p}{2})}{\Gamma(\frac{k+d+p}{2})} 
 \sim_{d,p} (-1)^{k/2+1} \sin\left(\frac{\pi p}{2}\right) \frac{1}{k^{d/2+p}}. \qedhere
\]
\end{proof}

\paragraph{Volume of Spherical Caps} For a point $x\in \bS^{d-1}$ and $r > 0$, consider the spherical cap $C(x,r) = \{y\in \bS^{d-1}: \norm{x - y}_2 \leq r \}$. It is clear that $\sigma_{d-1}(C(x,r))$ is independent of $x$ and so we write it as $\sigma_{d-1}(\scap(r))$. We cite a result on the volume of spherical caps from~\cite{boroczky2003} as follows.
\begin{lemma}[{\cite[Lemma 3.1]{boroczky2003}}] \label{lem:spherical cap}
When $r^2 \leq 2(1-1/\sqrt{d+1})$, it holds that
\[
\kappa_d\frac{r^{d-1}}{1-r^2/2}\left(1 - \frac{r^2}{4}\right)^{\frac{d-1}{2}}\leq \sigma_{d-1}(\scap(r)) \leq \kappa_d\frac{r^{d-1}}{1-r^2/2}\left(1 - \frac{r^2}{4}\right)^{\frac{d-1}{2}},
\]
where $\kappa_d = \Gamma(\frac{d}{2})/(2\sqrt{\pi}\Gamma(\frac{d+1}{2}))$.
\end{lemma}

\paragraph{Partition of Sphere} We shall need the following partition lemma in~\cite{Mat96}.
\begin{lemma}[{\cite[Lemma 5]{Mat96}}]
\label{lem:partition}
Let $P$ be an $N$-point set in $\bS^{d - 1}$, and let $s \ge 2$ be a constant. There is an $O(N^{d-\frac{1}{d-1}})$-time deterministic algorithm which computes a collection of disjoint $s$-point subsets $P_1, \dots, P_t \subset P$, which together contain at least half the  points of $P$, and with the following properties: 
\begin{enumerate}[label=(\roman*)]
    \item  For every $x \in \bS^{d - 1}$, the hyperplane $\{y: \inner{x, y} = 0\}$ only intersects
at most $O(N^{1 - \frac{1}{d - 1}})$ sets among $P_i$. Here ``intersect'' means that there exist $x$ and $y$ in $P_i$ such that $x, y$ are on different sides of the hyperplane.

\item  Each $P_i$ has diameter at most $O(N^{-\frac{1}{d - 1}})$.
\end{enumerate}
\end{lemma}

\section{$\ell_p$-Subspace Sketch Lower Bound} 
\label{sec:lower_bound}

We first state an auxiliary lemma.
\begin{lemma}\label{lem:f-value}
Suppose that $c(d)\leq r<1$. Define a function $f:[-1,1]\to\R$ as
\[
f(t) = \frac{1 - r^4}{(1 + r^4 - 2r^2 t)^{d / 2}} + \frac{1 - r^4}{(1 + r^4 + 2r^2 t)^{d / 2}} .
\]
Then for $0\leq t\leq 1$, it holds that
\[
f(t) \leq \frac{2(1-r^2)}{(2r^2)^{d/2}} \left(\frac{1}{(1-t)^{d/2}} + 1\right).
\]
\end{lemma}
\begin{proof}
Let $q = 1 - r^2$. Then the first term of $f(t)$ is
\[
\frac{1 - r^4}{(1 + r^4 - 2r^2 t)^{d / 2}} \leq \frac{2q - q^2}{(2 - 2q + q^2 - 2(1 - q)t)^{d / 2}} \leq \frac{2q}{(2(1-q)(1-t))^{d/2}}
\]
and the second term of $f(t)$ is
\[
\frac{1 - r^4}{(1 + r^4 + 2r^2 t)^{d / 2}} \leq \frac{2q - q^2}{(2 - 2q + q^2)^{d / 2}} \leq \frac{2q}{(2(1-q))^{d/2}}.
\]
The claimed result follows.
\end{proof}

Our result is mainly based on the following lemma. 

\begin{lemma}
\label{lem:min_error}
Suppose that $p\in [1,\infty)\setminus 2\Z$ is a constant.
Let $A$ and $B$ be sets of $n\leq N$ points on $\bS^{d - 1}$. Suppose that $A$ and $B$ are symmetric around the origin and $\|A_i - B_j\|_2 \geq \eta = C_1 N^{-\frac{1}{d - 1}} $ for all $i, j$. Then we have
\[
\delta \equiv \sup_{x \in \bS^{d - 1}} \frac{1}{n}\Abs{ \norm{Ax}_p^p - \norm{Bx}_p^p } \ge c_2 N^{-\frac{d + 2p}{2(d - 1)}}.
\]
 In the statement above, $C_1 > 0$ is a constant that depends only on $d$ and $c_2 > 0$ is a constant that depends on $d$ and $p$ only.
\end{lemma}
\begin{proof}
The proof is inspired by the proof of \cite[Proposition 6.6]{bourgain89}, which uses spherical harmonics.
Let $h_A(x) = \frac{1}{n} \|Ax\|_p^p = \frac{1}{n} \sum_i \Abs{\inner{A_i,x}}^p$  and let $h_B(x)$ be defined similarly. We expand $h_A - h_B$ into 
\[
h_A - h_B = \sum_{k = 0}^{\infty} \sum_j \inner{h_A - h_B, Y_{k, j}} Y_{k, j}. 
\]
where $\inner{h, Y_{k, j}}$ denotes the inner product in $L_2(\sigma_d)$ for which   
\[
\inner{h, Y_{k, j}} = \int_{\bS^{d - 1}} h(x) Y_{k, j}(x) \, d\sigma_{d-1}(x) .
\]
It follows from Parseval's identity that
\[
\delta^2 \ge \|h_A-h_B\|_{L_2(\sigma_{d-1})}^2 = \sum_{k = 1}^{\infty} \sum_j \inner{h_A - h_B, Y_{k, j}}^2 .
\]
Now, from~\eqref{eq:1} and~\eqref{eq:3} we have for all $u,v\in \bS^{d-1}$  that
\begin{equation}
    \label{eq:4}
    \frac{1 - r^2}{(1 + r^2 - 2r\inner{u, v})^{d/2}} = \sum_{k = 0}^{\infty}  r^k \sum_{j} Y_{k, j}(u) Y_{k, j}(v).
\end{equation}
Combining with~\eqref{eq:funk-hecke} and using the fact that $A$ and $B$ are symmetrically distributed on $\bS^{d-1}$, we obtain that
\[
    \frac{1}{n}\sum_{i = 1}^{n} \frac{(1 - r^2)}{(1 + r^2 - 2r\inner{u, A_i})^{d/2}} = 1 + \sum_{\text{even }k\geq 2} r^k \lambda_k^{-1} \sum_j \inner{h_A, Y_{k,j}} Y_{k, j}(u) ,
\]
and
\[
    \frac{1}{n}\sum_{i = 1}^{n} \frac{(1 - r^2)}{(1 + r^2 - 2r\inner{u, B_i})^{d/2}} = 1 + \sum_{\text{even }k\geq 2} r^k \lambda_k^{-1} \sum_j \inner{h_B, Y_{k,j}} Y_{k, j}(u) .
\]
Hence,
\[
    \frac{1}{n}\sum_{i = 1}^{n} \frac{(1 - r^2)}{(1 + r^2 - 2r\inner{u, A_i})^{d/2}} - \frac{1}{n}\sum_{i = 1}^{n} \frac{(1 - r^2)}{(1 + r^2 - 2r\inner{u, B_i})^{d/2}} = \sum_{\text{even }k\geq 2}  r^k \lambda_k^{-1} \sum_j \inner{h_A - h_B, Y_{k,j}} Y_{k, j}(u) ,
\]
Integrating with respect to $u$ on $\bS^{d-1}$, we have that
\begin{equation}\label{eq:6}
\begin{aligned}
    &\quad\ \left\|\frac{1}{n}\sum_{i = 1}^{n} \frac{(1 - r^2)}{(1 + r^2 - 2r\inner{u, A_i})^{d/2}} - \frac{1}{n}\sum_{i = 1}^{n} \frac{(1 - r^2)}{(1 + r^2 - 2r\inner{u, B_i})^{d/2}}\right\|_{L_2(\sigma_{d-1})}^2 \\
    &= \sum_{\text{even }k\geq 2} r^{2k} \lambda_k^{-2} \sum_j \inner{h_A - h_B, Y_{k,j}}^2 \\
    &\le \delta^2 \max_{\text{even }k\geq 2}(r^{2k} \lambda_k^{-2})  ,
\end{aligned}
\end{equation}
%
The leftmost side of~\eqref{eq:6} equals (using \eqref{eq:4})
\[
\frac{1}{n^2}\left( \sum_{i,j} \frac{1-r^4}{(1+r^4-2r^2\inner{A_i,A_j})^{d/2}} + \sum_{i,j} \frac{1-r^4}{(1+r^4-2r^2\inner{B_i,B_j})^{d/2}} - 2\sum_{i,j} \frac{1-r^4}{(1+r^4-2r^2\inner{A_i,B_j})^{d/2}} \right).
\]
Let $I, J$ be the index sets of $A$ and $B$ such that the points in $A_I, B_J$ are on the same hemisphere. We can rewrite the expansion above as
\begin{equation}
\label{eq:f}
\frac{2}{n^2} \left(\sum_{i \in I, j \in I} f(\inner{A_i, A_j}) + \sum_{i \in J, j \in J}  f(\inner{B_i, B_j}) - 2\sum_{i \in I, j \in J} f(\inner{A_i, B_j}) \right),
\end{equation}
where $f(t)$ is as defined in Lemma~\ref{lem:f-value}.

We now choose $(1 - r^2)^{-d + 1} = N$. Suppose that $N$ is sufficiently large such that $1 - r^2 \leq c_4(d) < 1$, where $c_4(d)$ is a constant depending only on $d$. Next we turn to lower bounding the expression in~\eqref{eq:f}.

Considering the summands with $i = j$. We see that the first two terms inside the bracket of \eqref{eq:f} are at least
\[
\sum_{i \in I} f(\inner{A_i, A_i}) + \sum_{j \in J} f(\inner{B_j, B_j})\ge n\frac{1 - r^4}{(1 - r^2)^d} \ge  n(1 - r^2)^{-d + 1} = n N.
\]
Next we bound the cross terms. We first show that, for every $A_i$,
\begin{equation}
\label{eq:bound_one_x}
\sum_{j \in J} f(\inner{A_i, B_j}) \le \frac{1}{4} N.
\end{equation}
Without loss of generality, we assume $\inner{A_i, B_j} \ge 0$ for all $j$, as otherwise we can replace $B_j$ with $-B_j$. To achieve this, consider a fixed $i\in I$. Let $I_k = \{j\in J: \ 2^{k - 1}\eta \le \norm{A_i- B_j}_2 < 2^{k}\eta \}$ for $k = 1, 2, \dots, K$, where $K$ is the smallest positive integer such that $2^K\eta \geq 1/2$.
We also define $I_{K+1} = \{j\in J:  \norm{A_i- B_j}_2 \geq 2^{K} \eta  \}$.
Then, we have that
\[
\sum_{j \in J} f(\inner{A_i, B_j}) = \sum_{k = 1}^{K+1} S_k = \sum_{k = 1}^{K+1} \left(\sum_{j \in I_k} f(\inner{A_i, B_j})\right)\; .
\]

Next, we will bound $S_k$ individually. 
From Lemma~\ref{lem:spherical cap} we have that 
\[
|I_k| \le \frac{\sigma_{d-1}(\scap(2^k\eta))}{\sigma_{d-1}(\scap(\eta/2))} \le C_3(d)\cdot (2^k)^{d-1}.
\]
By Lemma~\ref{lem:f-value}, when $0\leq t\le 1 - 2^{2k-1}\eta^2$ and $k\leq K$, we have that
\begin{align*}
f(t) &\leq \frac{2(1-r^2)}{(2r^2)^{d/2}} \left(\frac{1}{(1-t)^{d/2}} + 1\right) \\
&\leq \frac{2(1-r^2)}{(2(1-c_4))^{d/2}} \left(\frac{1}{2^{(k-1/2)d}\eta^d} + 1\right) \\
&\leq \frac{C_4(d)}{2^{kd}}\frac{1-r^2}{\eta^d}.
\end{align*}
Also, when $0 \le t < 7/8$, we have that 
\begin{align*}
f(t) &\leq \frac{2(1-r^2)}{(2r^2)^{d/2}} \left(\frac{1}{(1-t)^{d/2}} + 1\right) \\
&\leq \frac{2(1-r^2)}{(2(1-c_4))^{d/2}} \left(\frac{1}{(1/8)^{d/2}} + 1\right) \\
&\leq C_4(d)(1-r^2).
\end{align*}

Consequently (noting that $\inner{A_i, B_j} = 1 - \frac{\|A_i - B_j\|_2^2}{2}$),
\begin{align*}
\sum_{j \in J} f(\inner{A_i, B_j}) &\leq \left(\sum_{k \le K} C_3(d)(2^k)^{d-1} \cdot C_4(d)\frac{1-r^2}{\eta^d} \frac{1}{2^{kd}}\right) + \frac{n}{2} \cdot C_4(d)(1-r^2) \\
& \leq C_5(d)\left(\frac{1-r^2}{\eta^d} + \frac{N}{2} (1-r^2)\right) \leq \frac{1}{4}N
\end{align*}
provided that
\[
\frac{1}{4}N\geq C_5(d)\left(\frac{1-r^2}{\eta^d} + \frac{N}{2}(1-r^2)\right) = C_5(d)\left( N^{-\frac{1}{d-1}}\eta^{-d} + \frac{1}{2}N^{1-\frac{1}{d-1}} \right),
\]
which holds whenever
\[
\eta \geq C_1(d) N^{-\frac{1}{d-1}}
\]
and $N$ is sufficiently large.
It follows from \eqref{eq:bound_one_x} that
\[
2\sum_{i \in I, j \in J} f(\inner{A_i, B_j}) \le \frac{1}{2}n N,
\]
which implies that the expression in~\eqref{eq:f} is at least
\[
\frac{2}{n^2}\left(n N - \frac{1}{2} n N \right) \geq 1.
\]
Plugging this back into \eqref{eq:6} yields
\[
\delta^2 \max_{\text{even }k}(r^{2k} \lambda_k^{-2}) \ge 1 .
\]
By Lemma~\ref{lem:lambda_k}, there exists a constant $c(d)$ such that
\[
c(d, p) \delta^2 \max_{\text{even }k}(r^{2k} k^{d + 2p}) \ge \Omega(1).
\]
The maximum is attained when $k \approx (d+2)/\ln(1/r) \sim \frac{d}{1 - r^2}$, from which we obtain that 
\[
\delta \ge c(d,p) N^{-\frac{d + 2p}{2(d - 1)}}  . \qedhere
\]
\end{proof}

Equipped with the lemma above, we are now ready to prove our lower bound. We need the following lemma on the size of intersecting families.
\begin{lemma}[{\cite[Theorem 30]{Blocki2014}}]\label{lem:intersecting_family}
Suppose that $0<\alpha^2<\beta<\alpha<1$ and $n$ is sufficiently large. There exists a family $\mathcal{S}$ of subsets of $[n]$ such that (i) $|S| = \alpha n$ for each $S\in \mathcal{S}$, (ii) $|S\cap T|\leq \beta n$ for every pair of distinct $S,T\in \mathcal{S}$ and (iii) $|\mathcal{S}| = \Omega(e^{(\beta-\alpha^2)^2 n})$.
\end{lemma}

Let $\eta = C_1(d) N^{-\frac{1}{d - 1}}$, where $C_1(d)$ is as in Lemma~\ref{lem:min_error}. It follows from Lemma~\ref{lem:spherical cap} that we can take $n = c(d) N \leq N/2$ points $p_1, \dots, p_n$ on some spherical cap $C(x, \sqrt{2-\eta})$ such that $\|p_i - p_j\|_2 \ge \eta $ for every $1 \le i < j \le n$. Since $p_i$'s are chosen from a spherical cap of radius $\sqrt{2-\eta}$, we also have that $\norm{p_i + p_j}_2\geq \eta$ for all pairs $i\neq j$. Let $S = \{p_1, ..., p_n\}$. Applying Lemma~\ref{lem:intersecting_family}, we can find $m = 2^{\Omega(n)}$ subsets $S_1,\dots,S_m$ of $S$ such that $|S_i| = n/3$ for each $i$ and $|S_i\cap S_j| \leq n/6$ for every pair $i\neq j$. 

Given the approximation parameter $\eps$, let $N = c(d) \eps^{-\frac{2(d - 1)}{d + 2p}}$ be such that $c_2 N^{-\frac{d + 2p}{2(d-1)}} = 18\eps$, where $c_2$ is as in Lemma~\ref{lem:min_error}. 

Now we consider the following communication problem: suppose that Alice has one of the subsets $S_1, S_2, \dots, S_m$ and she wants to send a message to Bob, who needs to determine which subset Alice has. We shall show that we can solve this problem if we have an $\ell_p$-subspace sketch data structure. Suppose that the subset Alice has is $S_t$ and $Q$ is a data structure such that $Q(x) = (1 \pm \eps) \|S_t x\|_p$ for all $x \in \bS^{d - 1}$. Then, Alice sends such a data structure $Q$ to Bob.  

On the one hand, we have for the subset $S_i = S_t$ that
\begin{equation}\label{eqn:the-one-hand}
\left|Q(x) - \|S_ix\|_p\right| \le n \eps, \quad \forall x \in \bS^{d-1}.    
\end{equation}
On the other hand, if $S_t \ne S_i$, we know from the construction of the subsets $S_i$ that $|S_i \setminus S_t| = |S_t \setminus S_i| \geq  n/6$. Applying Lemma~\ref{lem:min_error} to $S_i\cup(-S_i)$ and $S_t\cup(-S_t)$, we see that there exists an $x \in \bS^{d - 1}$ such that 
\[
\Abs{ \|S_ix\|_p - \|S_t x\|_p } \ge \frac{1}{6}n\cdot c_2 N^{-\frac{d+2p}{2(d-1)}} \geq 3 n\eps, 
\]
whence it follows that
\begin{equation}\label{eqn:the-other-hand}
\left|Q(x) - \|S_ix\|_p\right| \ge |\|S_i x\|_p - \|S_tx\|_p| - \left|\|S_tx\|_p - Q(x)\right| \geq 3 n \eps - n\eps = 2n \eps.
\end{equation}
Combining \eqref{eqn:the-one-hand} and \eqref{eqn:the-other-hand}, we immediately see that Bob can decide which of the two cases he is in by querying a sufficient number of points on $\bS^{d - 1}$, which implies a lower bound of $\Omega(\log m) = \Omega(n) = \Omega(\eps^{-\frac{2(d - 1)}{d + 2p}})$ bits of space.
\begin{theorem}
\label{thm:lower_bound}
Suppose that $p\in [1,\infty)\setminus 2\Z$ and $d$ are constants. Any data structure that solves the $\ell_p$-subspace sketch problem for dimension $d$ and accuracy parameter $\eps$ requires $\Omega(\eps^{-\frac{2(d - 1)}{d + 2p}})$ bits of space.
\end{theorem}

\begin{remark}
Our proof does not work when $p$ is an even integer, which is as expected. In this case, $\norm{Ax}_p^p$ is a polynomial and  the spherical harmonic series will be finite. The term $\max_k (r^{2k}\lambda_k^{-2})$ in \eqref{eq:6} will be a constant instead of a quantity depending on $N$.
\end{remark}
\section{$\ell_p$-Subspace Sketch Upper Bound}
\label{sec:upper_bound}

Our approach follows \cite{Mat96}, which deals with the case of $p=1$. The following is our key lemma, which is an analogue of \cite[Proposition 7]{Mat96} for a general $p$.

\begin{lemma}
\label{lem:coreset}
Let $P \in \R^d$ be an $N$-point set with each point having $\ell_2$ norm $O(1)$, and let $w\in \simplex{N-1}$ be an  associated weight vector. There is an $O(d^{3/2}N^d)$-time deterministic algorithm which computes a subset $P' \subset P$ of at most $\frac{7}{8}N$ points with a weight vector $w'\in \simplex{|P'|-1}$ such that with probability $1 - 1 / N$, for every $x \in \bS^{d - 1}$ 
\[
\left|\sum_i w_i |\inner{P_i, x}|^p - \sum_i w'_i |\inner{P'_i, x}|^p \right| = O(N^{-\frac{d +2p}{2(d - 1)}}\sqrt{\log N}) \; .
\]
\end{lemma}

Given matrix $A$ and error parameter $\eps$, we apply Lemma~\ref{lem:coreset} repeatedly as in~\cite{Mat96}, obtaining a sequence of subsets with $\frac{7}{8}N$, $(\frac{7}{8})^2 N$, $\ldots$ points, until a subset of size $\tilde{O}(\eps^{-\frac{2(d - 1)}{d + 2p}})$ is obtained. Note that the errors of the successive approximations form a geometric progression, and hence the final error will be $O(\eps)$. The following theorem follows immediately.

\begin{theorem}
\label{thm:coreset_additive}
Suppose that $A \in \R^{n \times d}$ satisfies that $\norm{A_i}_2 = O(1)$ for all $i$ and $p$ is a positive integer constant. There is an $O(d^{3/2}n^d)$-time deterministic algorithm which computes a matrix $B$ consisting of $m = O(\eps^{-\frac{2(d - 1)}{d + 2p}}\log^{\frac{d - 1}{d + 2p}}(\eps^{-1}))$ rows of $A$, and an associated weight vector $w\in \simplex{m-1}$, such that with high probability for every $x \in \bS^{d - 1}$,
\[
\left|\sum_{i=1}^m w_i |\inner{B_i, x}|^p - \frac{1}{n}\norm{Ax}_p^p \right| = O(\eps) \; .
\] 
\end{theorem}

To prove Lemma~\ref{lem:coreset}, we first give intuition, which is inspired by the proof of~\cite[Proposition 7]{Mat96}.  Given a point set $P$, from Lemma~\ref{lem:partition} we know that for at least half of the points of $P$, we can divide them into groups $P_1, P_2, \ldots, P_t$ which satisfy the corresponding conditions (for the points that are not on the sphere, we scale them to unit vectors when applying the partition lemma). That is, (i) each $P_i$ is within a small area and (ii) for every $x \in \bS^{d - 1}$, the hyperplane $\{y: \inner{x, y} = 0\}$ only intersects a small number of sets among $P_i$.

Given a query point $x \in \bS^{d - 1}$, let $H$ be the hyperplane $\{y: \inner{x, y} = 0\}$. For the subset $P_i$, we first consider the case that $H$ does not intersect $P_i$. In this case we have 
\[
\sum_{y \in P_i} |\inner{x, y}|^p = \sum_{y \in P_i} \inner{x, y}^p \;.
\]
The tensor product trick $\inner{x, y}^p  = \inner{x^{\otimes p}, y^{\otimes p}} $ implies that
\[
\sum_{y \in P_i} |\inner{x, y}|^p = \sum_{y \in P_i} \inner{x^{\otimes p}, y^{\otimes p}} = \inner{x^{\otimes p},\sum_{y \in P_i} y^{\otimes p}}\;.
\]
Hence, what we need to store is $\sum_{y \in P_i} y^{\otimes p}$, which can be seen as a $d^p$-dimensional vector.

For the other case where $H$ intersects $P_i$, the method above will not work. However, note that in this case $|\inner{x, y}|$ is small for every $y \in P_i$ as each $P_i$ lies in a small region. We also know that $H$ intersects only a small number of groups $P_i$; therefore, sampling points from the intersection is enough to achieve an $\eps$-additive error.

However, one issue here is that it is not easy to determine whether $H$ intersects $P_i$ if we only store a limited number of points. To address this, for each subset $P_i$, we choose a random subset $T\subseteq P_i$ with associated weights $\{w_y\}_{y\in T}$ such that (i) $\sum_{y \in T} w_y \cdot y^{\otimes p} = \sum_{y \in P_i} y^{\otimes p}$ (which is helpful for the first case) and (ii) $\E \sum_{y \in T} w_y \cdot |\inner{x, y}|^p = \sum_{y \in P_i} \inner{x, y}^p$ (which we will show is enough for the second case). The following lemma will be useful.

\begin{lemma}[\cite{Mat96}, Lemma~8]
\label{lem:weights}
Let $P = \{P_1,\dots,P_s\} \subset \R^d$ be a set of $s\ge d + 1$ points with an associated weight vector $u\in \simplex{s-1}$. There is a deterministic algorithm which computes a group of the subsets $T_1, \dots, T_{s'}$, each associated with a probability $p_i$ and a weight vector $w_i\in \simplex{t_i-1}$, where $s' \le s - d$ and $t_i = |T_i|$, such that 
\begin{enumerate}[label=(\roman*)]
    \item $\sum_{i=1}^{s'} p_i = 1$;
    \item $t_i \leq d + 1$ for each $i \in [s']$;
    \item $\sum_{j=1}^{t_i} w_{i,j} T_{i,j} = \sum_{i=1}^s u_i P_i$ for each $T_i = \{T_{i,j}\}_{j=1,\dots,t_i}$;
    \item $\sum_{i \in I_k} p_i w_{i,j(i,k)} = u_k$ for each $k \in [s]$, where $I_k = \{i\in [s']: P_k\in T_i\}$ and $j(i,k)$ is the index $j$ such that $T_{i,j}=P_k$ for $i\in I_k$.
\end{enumerate}
Furthermore, the running time of the algorithm is dominated by that of finding a feasible solution to a linear program with $s$ variables and $d+1$ constraints.
\end{lemma}
We now are ready to prove our Lemma~\ref{lem:coreset}.
\begin{proof}[Proof of Lemma~\ref{lem:coreset}]
Since $\sum_i u_i = 1$, there are at least $\frac{1}{2}N$ points with weight $w_i \le \frac{2}{N}$. Applying Lemma~\ref{lem:partition} to these points with $s = 2(d^p + 1)$ (recall that $d$ and $p$ are both constants), we obtain a collection of disjoint $s$-points subsets $P_1, ... ,P_t$, each of which has diameter $O(N^{\frac{1}{d - 1}})$. Furthermore, for every $x \in \bS^{d - 1}$, the corresponding hyperplane $\{y: \inner{x, y} = 0\}$ intersects at most $O(N^{1 - \frac{1}{d - 1}})$ subsets. For the remaining points, we keep them with the same weights. For each $P_i$, we sample at most half of the points for each group as below. 

For each $x \in \R^d$, let $T(x) = x^{\otimes p}$ denote its $p$-th tensor product, viewed as a $d^p$-dimensional vector. For each subset $P_i$, let $T(P_i) = \{T(x): x \in P_i\}$ and $w(P_i) = \sum_{y \in P_i} w_y$. We then apply Lemma~\ref{lem:weights} to $T(P)$ with the normalized weights $w'_y = w_y / w(P_i)$, obtaining a group of the subsets $T_1, \dots, T_{s'}$ with weights $v_1, \dots, v_{s'}$ and probabilities $p_1 \dots ,p_s'$. We choose a random index $j \in \{1, 2, \dots, s'\}$ according to probabilities $p_1, \dots, p_s'$ and take the subset $T_j$ to be our sample set of the points. The guarantee of Lemma~\ref{lem:weights} implies that $T_j$ contains at most $d^p + 1$ points, which is at most half of the size of $P_i$. Repeating this procedure for each $P_i$, we finally obtain a subset of $P$ containing at most $\frac{7}{8}N$ points.

Now we analyze correctness of our algorithm. Fix $x \in \bS^{d - 1}$. For each $P_i$, let $Q_i$ denote the subsets of $P_i$ with the points we sampled in the procedure above with the associated weight $v_i$. It suffices to show that
\[
\left|\sum_i \sum_{y \in Q_i} v_y |\inner{x, y}|^p - \sum_i \sum_{y \in P_i} w_y |\inner{x, y}|^p\right| = O(N^{-\frac{d +2p}{2(d - 1)}}\sqrt{\log N}) 
\]
holds with high probability.

Let $I$ denote the set of indices in $\{1, 2, \ldots ,t\}$ for which the hyperplane $H = \{y: \inner{x, y} = 0\}$ intersects $P_i$. Then for each $i \notin I$, from the condition 
\[
\sum_{y \in Q_i }v_y y^{\otimes p} = \sum_{y \in P_i} w_y y^{\otimes p}
\]
we have that
\[
\inner{x^{\otimes p}, \sum_{y \in Q_i} v_y y^{\otimes p}} = \sum_{y \in Q_i} v_y \inner{x, y}^p =  \sum_{y \in P_i} w_y \inner{x, y}^p = \inner{x^{\otimes p}, \sum_{y \in P_i} w_y y^{\otimes p}} \;.
\]
Note that in this case $\inner{x, y}$ has the same sign for all $y \in P_i$. Hence
\[
\sum_{y \in Q_i} v_y |\inner{x, y}|^p =  \sum_{y \in P_i} w_y |\inner{x, y}|^p
\]
Now we focus on the case $i \in I$. Recall that $|I| = O(N^{1 - \frac{1}{d - 1}})$. Define a random variable $X_i$ as
\[
X_i = \sum_{y \in Q_i} v_y |\inner{x, y}|^p -  \sum_{y \in P_i} w_y |\inner{x, y}|^p \;,
\]
where the randomness is taken over the choice of the subsets $Q_i$. Lemma \ref{lem:weights}(iii) implies that 
\[
\E\left[\sum_{y \in Q_i} v_y |\inner{x, y}|^p\right] = \sum_{y \in P_i} w_y |\inner{x, y}|^p \;.
\]
Hence $\E X_i = 0$. Since $H$ intersects $P_i$ and the diameter of $P_i$ is $O(N^{-\frac{1}{d-1}})$, it holds that $|\inner{x, y}|^p = O(N^{-\frac{p}{d - 1}})$ for all $y \in Q_i$. Recalling our definitions of $v_y$ and $s$, we have $\sum_{y\in Q_y} {v_y} = w(P_i) = O(s/N)$ and thus
$|X_i| = O(N^{-1 - \frac{p}{d - 1}})$. Let $U$ be this upper bound for $|X_i|$. It then follows from Bernstein's inequality that
\[
\Pr\left[\left|\sum_{i \in I}X_i\right| > \lambda U \sqrt{|I|}\right] < 2e^{-\lambda^2 / 2},
\]
where $U\sqrt{|I|} = O(N^{-\frac{d +2p}{2(d - 1)}})$. Taking $\lambda \sim \sqrt{\log N}$, we obtain that for a fixed $x$,
\begin{equation}\label{eqn:fixed_x_bound}
\left|\sum_i \sum_{y \in Q_i} v_y |\inner{x, y}|^p - \sum_i \sum_{y \in P_i} w_y |\inner{x, y}|^p\right| = O(N^{-\frac{d +2p}{2(d - 1)}}\sqrt{\log N})
\end{equation}
with probability at least $1 - 1/\poly(N)$.

Next we do a net argument. Choose a constant $D > \frac{d+2p}{2(d-1)}$ and let $\gamma = N^{-D}$. We can take a $\gamma$-net $\mathcal{N}$ on $\bS^{d - 1}$ with $O(\gamma^{-(d - 1)})$ points. By a union bound, we have that with probability at least $1 - 1/N$, the bound~\eqref{eqn:fixed_x_bound} holds for all $x\in\mathcal{N}$ simultaneously. Note that for two $x, x' \in \bS^{d - 1}$, if $\norm{x - x'}_2 \le \gamma$, then $\abs{|\inner{x, y}|^p - |\inner{x', y}|^p} = O(\gamma)$. It follows that with probability at least $1 - 1/N$, the target bound \eqref{eqn:fixed_x_bound} holds for all $x \in \bS^{d - 1}$ simultaneously, which is what we need.

We now analyze the time complexity of our algorithm. The first step is to compute the subsets $P_1, \dots, P_t$, which can be done in $O(N^{d-\frac{1}{d-1}})$ time by Lemma~\ref{lem:partition}. Then, for each subset $P_i$, by Lemma~\ref{lem:weights}, the subset $T$ and the associated weights $w$ can be computed in $\tilde{O}(d^{\frac{3}{2}}s) = \tilde{O}(d^{\frac{3}{2}+p})$ time (for instance, using the LP algorithm in \cite{LS15}). Therefore, the total runtime of the algorithm is $O(d^{3/2}N^{d})$, as claimed.
\end{proof}

\paragraph{Achieving $(1 \pm \eps)$-Multiplicative Error.} Consider a general matrix $A \in \R^{n \times d}$. Without loss of generality, assume that $A$ has full column rank. 

Consider the John ellipsoid of $Z(A) := \{x\in \R^d: \norm{Ax}_p \leq 1\}$. There exists an invertible linear transformation $T:\R^d\to \R^d$ such that $B_{\ell_2}^d \subseteq Z(AT)\subseteq \sqrt{d}B_{\ell_2}^d$. Let $A'=AT$. Then $1/\sqrt{d} \leq \norm{A'x}_p \leq 1$ when $x\in \bS^{d-1}$. Let $A''_i = A'_i/\norm{A'_i}_2$. Then $A''_i \in \bS^{d-1}$. Note that
\[
1\geq \norm{A'x}_p^p = \sum_i \norm{A'_i}_2^p \abs{\inner{A''_i, x}}^p,\quad \forall x\in \bS^{d-1}.
\]
Integrate both sides over $\bS^{d-1}$ w.r.t.\ $x$. Observe that $\int_{\bS^{d-1}} \abs{\inner{u,x}}^p d\sigma_{d-1}(x)$ is a constant, independent of $u$, whenever $u\in \bS^{d-1}$. Denote this constant by $\beta_{d,p}$. It follows that
\[
1 \geq \sum_i \norm{A'_i}_2^p \cdot \beta_{d,p},
\]
and thus $\norm{A'x}_p^p \geq 1/d^{p/2}\geq  (\beta_{d,p}/d^{p/2}) \sum_i \norm{A'_i}_2^p$ for all $x\in \bS^{d-1}$.

Next, define weights $w_i' = w_i \norm{A_i'}_2^p$. Then
\[
\sum_i w_i \abs{\inner{A'_i,x}}^p = \sum_i w_i' \abs{\inner{A''_i,x}}^p.
\]
Suppose that $w_i = 1/n$ for all $i$. 
We apply Lemma~\ref{lem:coreset} to $A''$ with weights $w''_i = w_i'/\sum_j w_j'$ and 
obtain a subset of points $W$ with weights $v_i$, where $|W| = O(\eps^{-\frac{2(d - 1)}{d + 2p}})$, such that with high probability,  
\begin{equation}\label{eqn:intermediate_bound_normalized}
\left|\sum_i v_i |\inner{W_i, x}|^p - \frac{1}{n \sum_j w_j'} \norm{A'x}_p^p \right| = O(\eps), \quad \forall x \in \bS^{d - 1}.
\end{equation}
Recall that we showed above that  $\norm{A'x}_p^p \geq c(d,p)\cdot n\sum_j w_j' = c(d,p)\sum_i \norm{A'_i}_2^p$ for all $x\in \bS^{d-1}$, where $c(d,p) > 0$ is some constant depending only on $d$ and $p$. 

Let $v_i' = n v_i \sum_j w_j'$. It follows from~\eqref{eqn:intermediate_bound_normalized} that 
\[
\left|\sum_i v_i' |\inner{W_i, x}|^p - \norm{A' x}_p^p \right| = O(\eps)\norm{A' x}_p^p, \quad \forall x \in \bS^{d - 1}.
\]
Therefore, 
\[
\left|\sum_i v_i' |\inner{(T^{-1})^\top W_i, Tx}|^p - \norm{ATx}_p^p \right| = O(\eps)\cdot \norm{ATx}_p^p, \quad \forall x\in\R^d.
\]
This implies that the rows of $W T^{-1}$ (which form a row-subsampled matrix of $A$) with weights $v_i'$ are exactly what we need. Our final theorem is now immediate.

\begin{theorem}
\label{thm:coreset_multiplicative}
Suppose that $A \in \R^{n \times d}$ and $p$ is a positive integer constant. There is a polynomial-time algorithm which computes a subset $B$ of $m = \tilde{O}(\eps^{-\frac{2(d - 1)}{d + 2p}})$ rows of $A$ and an associated weight vector $w\in\R^m$, such that with high probability for every $x \in \bS^{d - 1}$,
\[
\left|\sum_{i=1}^m w_i |\inner{B_i, x}|^p - \norm{Ax}_p^p \right| = O(\eps) \cdot \norm{Ax}_p^p \; .
\] 
\end{theorem}
\begin{proof}
First, we show that a constant-factor approximation of the John ellipsoid of $Z(A)$ can be found in polynomial time. We can compute in polynomial time a decomposition $A=UT$, where $U\in \R^{n\times d}$ has orthonormal columns and $T\in\R^{d\times d}$ is an invertible matrix. It then suffices to find the John ellipsoid of $Z(U)$. It is clear that $B(0,r)\subseteq Z(U)\subseteq B(0,R)$ for $r=1/n$ and $R=n^{\max\{1/2-1/p,0\}}$ and that a separation oracle for $Z(U)$ can be implemented in polynomial time. Then an ellipsoid $E$ satisfying $E\subseteq Z(U)\subseteq \sqrt{d(d+1)}E$ can be found in polynomial time via the shallow-cut ellipsoid method~\cite[Theorem 4.6.3]{GLS}. Therefore,  $T^{-1}E \subseteq Z(A)\subseteq \sqrt{d(d+1)} T^{-1}E$, that is, $T^{-1}E$ is a constant-factor approximation of the John ellipsoid of $Z(A)$.

Correctness then follows from the discussion preceding the theorem statement (which goes through with a constant-factor approximation of John's ellipsoid) and Theorem~\ref{thm:coreset_additive}, which also implies that the remaining procedure can be completed in polynomial time.
\end{proof}

\section{$\ell_p$-Subspace Sketch in a One-Pass Stream} 
\label{sec:stream}
In this section, we implement our algorithm in the previous section in a one-pass stream, where each row arrives one at a time. We assume that each entry of matrix $A$ can be saved in $\log(n)$ bits of space. We first show that a coreset can be constructed in linear space when the number of rows is not large, which will be used as a subroutine in our full algorithm. Then we present the full algorithm, which is based on the standard merge-and-reduce paradigm and uses an additional factor of $\polylog(n)$ space. Finally, we show that the $\polylog(n)$ factor can be reduced to $\poly(\log\log n)$. At the end of this section, we shall give another $O(\eps)$-additive error streaming algorithm with $O(1)$ update time, with a slightly worse space complexity $\tilde{O}(\eps^{-\frac{2(d - 1)}{d + 2p - 1}})$ for general $d = O(1)$ and the near-tight bound $\tilde{O}(\eps^{-\frac{2(d - 1)}{d + 2p}})$ for $d \le 2p + 2$.

\subsection{Basic Step: Coreset}
Our algorithm is based on the following lemma.

\begin{lemma}
\label{lem:coreset_space}
Suppose that $p$ is a positive integer constant, $A \in \R^{n \times d}$ and $w\in \simplex{n-1}$ is the associated weight vector. There is an algorithm \textsc{Coreset($A, w, \eps$)} which computes a subset $B$ of $m = \tilde{O}(\eps^{-\frac{2(d - 1)}{d + 2p}})$ rows of $A$ and a weight vector $v\in \simplex{m-1}$, such that with high probability, it holds for every $x \in \bS^{d - 1}$,
\[
\left|\sum_i v_i |\inner{B_i, x}|^p - \sum_i w_i |\inner{A_i, x}|^p \right| = O(\eps) \cdot \left(\sum_i w_i |\inner{A_i, x}|^p \right) \; .
\] 
Furthermore, the algorithm can be implemented in $O(n)$ space.
\end{lemma}

\begin{proof}
The ellipsoid method employed in the proof of Theorem~\ref{thm:coreset_multiplicative} operates in $\R^d$ with $\poly(d)$ space, except for the separation oracle which can be clearly implemented in $O(n)$ space. Therefore, the linear transformation that normalizes the John ellipsoid of $Z(A)$ can be computed in $O(n)$ space. Recall that we only need a constant-factor approximation. Hence,  we can assume that after the linear transformation, each entry can still be stored in $O(\log n)$ bits of space.

Since we only need to invoke Lemma~\ref{lem:coreset} a total of $\log n$ times, it suffices to show that $O(N)$ space is enough in each invocation, where $N$ is the number of rows of the input. From the proof of Lemma~\ref{lem:coreset}, there exist desired groups $P_1,\dots,P_t$, which can be found in $O(N)$ space by enumeration. Each subset $P_i$ corresponds to only $O(s)$ rows, and thus the weights can be computed in $\poly(s)$ space. Thus, each run of Lemma~\ref{lem:coreset} can be implemented in $O(N)$ space. This finishes the proof.
\end{proof}

\subsection{Merge and Reduce}
Given the coreset procedure (Lemma~\ref{lem:coreset_space}), the general coreset framework in \cite{BDM+20} can be readily applied, leading to an algorithm similar to Algorithm 6 therein. For completeness, a full description of our algorithm is presented in Algorithm~\ref{alg:1}. 
We maintain a number of blocks $\B_0,\B_1,\dots$, each of size $\ms$ (which is determined by the coreset size). The most recent rows are stored in $\B_0$; whenever $\B_0$ is full, the successive non-empty blocks $\B_0,\ldots,\B_i$ are merged and reduced to a new coreset, which will be stored in $\B_{i+1}$. Since there are $n$ data points in the stream, the next lemma guarantees that maintaining $(\log n + 1)$ blocks $\B_0,\dots,\B_{\log n}$ suffices. 

\begin{algorithm}[h]
    \caption{Merge-and-reduce framework for the $\ell_p$ subspace sketch for constant $d$.}
    \label{alg:1}
    \textbf{Input:} A stream of rows $a_1, a_2, \ldots, a_n \in \R^{1 \times d}$ with weights $u_1, \ldots, u_n$, and approximation factor $\eps$\;
    \textbf{Output:} A coreset $B$ with associated weights $w$ \;
    \SetAlgoLined
    Initialize blocks $\B_0,\B_1,\dots,\B_{\log n}\gets\emptyset$, $\ms \gets \tilde{O}(\gamma^{-\frac{2(d - 1)}{d + 2p}})$ where $\gamma = \frac{\eps}{\log n}$\;
    \ForEach{row $a_t$ and weight $u_t$}{
    \uIf{$\B_0$ does not contain $\ms$ rows}{
        $\B_0\gets a_t\circ\B_0$\;
    }
\Else{
Let $i>0$ be the minimal index such that $\B_i=\emptyset$\;

$\B_i, w_i \gets\textsc{Coreset}\left(\textbf{M}, w,\frac{\eps}{\log n}\right)$, where $\textbf{M}=\B_0\circ\dots\circ\B_{i-1}$ and $w=w_0\circ \dots \circ w_{i-1}$\;

\For{$j=0$ \KwTo $j=i-1$}{
$\B_j\gets\emptyset$ \;
}
$\B_0 \gets a_t$, $w_0 \gets u_t $\;
}
}
$\B^\ast, w^\ast \gets \textsc{Coreset}(\B_{\log n}\circ\dots\circ\B_0, w_{\log n} \circ \dots \circ w_{0}, \eps)$\;
\KwRet $\B^\ast$ and $w^\ast$
\end{algorithm}

\begin{lemma}[\cite{BDM+20}]
\label{lem:merge_and_reduce}
Suppose that $\B_0,\dots,\B_{i-1}$ are all empty while $\B_i$ is non-empty. Then $\B_i$ with the associated weight vector $w_i$ is a $(1 + \frac{\eps}{\log n})^i$-coreset for the last $2^{i - 1} \ms$ rows. 
\end{lemma}
\begin{theorem}
\label{thm:merge_and_reduce}
Let $A = a_1 \circ \cdots \circ a_n$ be a stream of $n$ rows, where $a_i \in \R^{1 \times d}$. There is an algorithm which computes a subset $B$ of $m = \tilde{O}(\eps^{-\frac{2(d - 1)}{d + 2p}})$ rows of $A$ and an associated weight vector $w\in \R^m$, such that with high probability, for every $x \in \bS^{d - 1}$,
\[
\left|\sum_i w_i |\inner{B_i, x}|^p - \norm{Ax}_p^p \right| = O(\eps) \cdot  \norm{Ax}_p^p \; .
\] 
Moreover, the algorithm can be implemented in $\tilde{O}(\eps^{-\frac{2(d - 1)}{d + 2p}}\cdot \polylog(n))$ words of space.
\end{theorem}
\begin{proof}
We assume that each call to the subroutine $\textsc{Coreset}(\textbf{M}, w, \frac{\eps}{\log n})$ is successful, which holds with high probability after taking a union bound. It then follows from Lemma~\ref{lem:merge_and_reduce}  that $\B_{\log n} \circ \cdots \circ \B_0$ is a $(1 + \frac{\eps}{\log n})^{\log n} = (1 + O(\eps))$-coreset of the rows of $A$ and, after a further coreset operation, $\B^\ast$ is a $(1+O(\eps))(1+\eps) = (1+O(\eps))$-coreset of $A$, as desired.

Now we analyze the space complexity of our algorithm. Let $\gamma = \frac{\eps}{\log n}$. The algorithm stores $O(\log n)$ blocks $\B_i$, each taking at most $O(\ms\cdot d) = \tilde{O}(\gamma^{-\frac{2(d - 1)}{d + 2p}})$ words of space. Lemma~\ref{lem:coreset_space} implies that each call to the subroutine $\textsc{Coreset}(\textbf{M}, w, \gamma)$ takes $O(\ms \cdot \log n)$ words of space, since the number of total rows in the input does not exceed $O(\ms \cdot \log n)$. Therefore, the total space of our algorithm is $O(\ms\cdot \log n) = O(\eps^{-\frac{2(d - 1)}{d + 2p}}\cdot \polylog(n/\eps))$.
\end{proof}

\subsection{Reducing Space Complexity} 
In this section, we show that the multiplicative $\poly(\log n)$ factor in the space complexity can be reduced to $\poly(\log\log n)$ at the cost of an extra additive $\poly(\log n)$ term. The basic idea is that if we can sample a few of the rows which form a good approximation to the matrix $A$, then we can then treat the sampled rows as a new stream and the new input to our Algorithm~\ref{alg:1}. To this end, we consider the following $\ell_p$ sensitivities of the rows of $A$.

\begin{definition}
For a matrix $A = a_1 \circ \cdots \circ a_n \in \R^{n \times d}$, the $\ell_p$-sensitivity of $a_i$, denoted by $s_i(A)$, is defined to be $s_i(A) = \sup_{x \in \R^d\setminus\{0\}} \frac{|\inner{a_i, x}|}{\|Ax\|_p^p}$.
\end{definition}

The following lemma shows that sampling according to the $\ell_p$-sensitivities of $A$ gives a subspace embedding of $A$. This is a generalization of the $\ell_1$-sensitivity sampling~\cite[Lemma 4.4]{BDM+20} and the proof follows the same approach.

\begin{lemma}
\label{lem:sensitivity}
Let $A \in \R^{n \times d}$ and $1 \le p < \infty$. The matrix $B$ is a submatrix of $A$ such that the rescaled $i$-th row $p_i^{-1/p}a_i$ is included in $B$ with probability $p_i \ge \min(\beta s_i(A), 1)$. Then, there is a constant $c$ such that when $\beta \ge c \eps^{-2} d \log(1/\eps)$,  the matrix $B$ is a $(1 \pm \eps)$-subspace embedding of $A$ with probability at least $9/10$.
\end{lemma}
\begin{proof}
Since the row space of $B$ is contained in that of $A$, in order to show that $B$ is a $(1\pm\eps)$-subspace embedding of $A$, it suffices to show that 
$\norm{Bx}_p = 1\pm\eps$ for all $x$ such that $\norm{Ax}_p = 1$.

Fix an $x$ and let $y = Ax$. Define the random variable $Z_i$ to be $|y_i|^p/p_i$ with probability $p_i$ and $0$ otherwise, so $\E Z_i = |y_i|^p$. Let $Z = \norm{Bx}_p^p$.  Then $Z = \sum_i Z_i$ and
\[
\E[Z] = \sum_i \E[Z_i] = \norm{y}_p^p\;. 
\]
We also have
\[
\Var[Z_i] \le \E[Z_i^2] = \frac{|y_i|^{2p}}{p_i^2} \cdot p_i = \frac{|y_i|^{2p}}{p_i} 
\]
and
\[
Z_i \le \frac{|y_i|^p}{p_i} \le  \frac{1}{\beta}\frac{|y_i|^p \norm{y}_p^p}{|y_i|^p} = \frac{\norm{y}_p^p}{\beta}\;,
\]
so
\[
\Var[Z] = \sum_i \Var[Z_i] \le \sum_{i = 1}^n \frac{|y_i|^{2p}}{p_i} \le \frac{1}{\beta} \sum_{i = 1}^n |y_i|^p \frac{|y_i|^p \norm{y}_p^p}{|y_i|^p} = \frac{\norm{y}_p^{2p}}{\beta}\;.
\]
It follows from Bernstein's inequality that
\[
\Pr\left[ \Abs{Z - \E[Z]} \ge \eps \norm{y}_p^p \right] \le 2\exp\left(-\frac{\beta}{2}\frac{\eps^2 \norm{y}_p^{2p}}{\norm{y}_p^{2p} + \eps \norm{y}_p^{2p} / 3}\right) = \eps^{-\Omega(d)}\;.
\]
Rescaling $\eps$ by a constant factor (depending on $p$), we have that $\norm{Bx}_p = (1\pm\eps)\norm{Ax}_p$ with probability at least $1 - \eps^{\Omega(d)}$ for each fixed $x$.

We next need a net argument. Let $\mathcal{S} = \{Ax: x \in \R^d, \norm{Ax}_p = 1\}$ be the unit $\ell_p$ ball and $\mathcal{N}$ be a net of size $(3/\eps)^{d}$ of $\mathcal{S}$ under the $\ell_p$ distance. By a union bound, we have that $\norm{Bx}_p = (1 \pm \eps) \norm{Ax}_p$ for every $Ax \in \mathcal{N}$ simultaneously with probability at least $9/10$. Conditioned on this event, for each $y = Ax \in \mathcal{S}$, we choose a sequence of points $y_0,y_1,\dots\in \mathcal{S}$ as follows.
\begin{itemize}
    \item Choose $y_0 \in \mathcal{S}$ such that $\norm{y - y_0}_p \le \eps$ and let $\alpha_0 = 1$;
    \item After choosing $y_0,y_1,\dots,y_i$, we choose $y_{i+1}$ such that
    \[
    \norm{ \frac{y - \alpha_0 y_0 - \alpha_1 y_1 - \cdots - \alpha_i y_i}{\alpha_{i+1}} - y_{i+1} }_p \leq \eps,
    \]
    where $\alpha_{i+1} = \norm{y - \alpha_0 y_0 - \alpha_1 y_1 - \cdots - \alpha_i y_i}_p$. 
\end{itemize}
The choice of $y_{i+1}$ means that
\[
    \alpha_{i+2} = \norm{ y - \alpha_0 y_0 - \alpha_1 y_1 - \cdots - \alpha_i y_i - \alpha_{i+1}y_{i+1} }_p \leq \alpha_{i+1}\eps.
\]
A simple induction yields that $\alpha_i \leq \eps^i$. Hence
\[
    y = y_0 + \sum_{i \ge 1} \alpha_i y_i , \ \ \ |\alpha_i| \le \eps^{i} \;.
\]
Suppose that $y_i = Ax_i$, we have
\[
\norm{Bx}_p \le \norm{Bx_0}_p + \sum_{i \ge 1} \epsilon^i \norm{Bx_i}_p \le (1 + \epsilon) + \sum_{i \ge 1} \epsilon^i(1 + \epsilon) = 1 + O(\epsilon),
\]
and 
\[
\norm{Bx}_p \ge \norm{Bx_0}_p - \sum_{i \ge 1} \epsilon^i \norm{Bx_i}_p \ge (1 - \epsilon) - \sum_{i \ge 1} \epsilon^i(1 - \epsilon) = 1 - O(\epsilon).
\]
Rescaling $\eps$ again gives the result.
\end{proof}

Since the rows of $A$ are given in the  streaming model, we shall sample according to the online $\ell_p$ sensitivities of the rows, which are defined below.
\begin{definition} [Online $\ell_p$ Sensitivities \cite{BDM+20, WY22}]
    Let $A \in \R^{n \times d}$ and let $1 \le p < \infty$. Then, for each $i\in [n]$, the $i$-th online $\ell_p$ sensitivity is defined as
    \[
        s_i^{\OL}(A) \coloneqq \begin{cases}
            \min \left\{\sup_{x \in \rowspan(A) \setminus \{0\}} \frac{|\inner{a_i,x}|^p}{\norm{A_{i - 1}x}_p^p}, 1\right\} &  a_i \in \rowspan(A_{i - 1})\\
            1 & \text{otherwise},
        \end{cases}
    \]
    where $A_j\in\mathbb R^{j\times d}$ denotes the submatrix of $A$ formed by the first $j$ rows.

\end{definition}
It is clear from the definition that the online $\ell_p$ sensitivity is at least as large as the $\ell_p$-sensitivity of the same row and we can thus use the online $\ell_p$-sensitivities in an online algorithm to achieve the $\ell_p$-subspace embedding property. 
Our full algorithm is given in Algorithm~\ref{alg:2}.

\begin{algorithm}[h]
    \caption{Algorithm for $\ell_p$ subspace sketch for constant $d$.}
    \label{alg:2}
    \textbf{Input:} A streaming of rows $a_1, a_2, \ldots, a_n \in \R^{1 \times d}$ and approximation factor $\gamma$\;
    \SetAlgoLined
    $\mathsf{ALG_1}$ is an instance of Algorithm~\ref{alg:1} with $\eps = O(1)$\;
    $\mathsf{ALG_2}$ is an instance of Algorithm~\ref{alg:1} with $\eps = \gamma$\;
    $\beta \gets \poly(d) \log(1/\eps)/\eps^2$\;
    \ForEach{row $a_t$}{
    $M_t \gets$ a coreset of $A_{t - 1}$ from $\mathsf{ALG}_1$\;
    $\tau_t \gets$ a $\poly(d)$-approximation of $s_t^{\OL}(A)$ from $a_t$ and $M_t$\; 
    $p_t \gets \min(1, \beta\tau_t)$\;
    With probability $p_t$, feed $a_t / p_t^{1/p}$ to $\mathsf{ALG}_2$\;
    Feed $a_t$ to $\mathsf{ALG}_1$\;
}
\KwRet a coreset $M$ from $\mathsf{ALG}_2$.

\end{algorithm}

It follows from Theorem~\ref{thm:merge_and_reduce} and Lemma~\ref{lem:sensitivity} that the output $M$ is a $(1 \pm \eps)$-coreset of $A$ with high probability. To bound the space complexity of our algorithm, we need to bound the sum of the online $\ell_p$ sensitivities of the rows $a_i$, which is given below in Lemma~\ref{lem:bound_of_sensitivity}. It follows that the expected number of sampled rows is $O(\eps^{-2}\poly(d) \log(1/\eps)\log n) = O(\eps^{-2} \log(1/\eps)\log n)$.

\begin{lemma}[{Lemma~4.7 in~\cite{BDM+20}, Theorem~3.10 in~\cite{WY22}}]
\label{lem:bound_of_sensitivity}
Suppose that $A \in \R^{n \times d}$ and $p \in \{1\}\cup[2,\infty)$. Let $q = 1$ when $p = 1$ or $q = p/2$ when $p \ge 2$, and $\kappa$ be the condition number of $A$. Then, we have
\[
\sum_{i = 1}^n s_i^{\OL}(A) = O(d^{q}(\log n) \log^{q} \kappa).
\]
Moreover, if $A \in \Z^{n \times d}$ is the integer matrix with entries bounded by $\poly(n)$, we have 
\[
\sum_{i = 1}^n s_i^{\OL}(A) = O(d^{q}\log^{q + 1} n).
\]
\end{lemma}

The only thing remaining is to compute a $\poly(d)$-approximation of $s_i$ from $a_i$ and $M_i$. Since $\norm{A_{i-1}x}_p = (1\pm \eps)\norm{M_i x}_p$ for all $x$, it must hold that $(\rowspan(A_{i-1}))^\perp = (\rowspan(M_i))^\perp$ and so $\rowspan(A_{i-1}) = \rowspan(M_i)$. Hence, we can use $M_i$ to determine whether $a_i \in \rowspan(A_{i - 1})$. When $a_i \in \rowspan(A_{i - 1})$, we need an efficient algorithm to find $\tau_t$, for which we consider a well-conditioned basis of $A_{i - 1}$, defined below.

\begin{definition}[Well-conditioned basis, \cite{DDH+09}] 
Suppose that $A \in \R^{n \times d}$ has rank $r$ and $p\geq 1$. An $ n \times r$ matrix $U$ is an $(\alpha, \beta, p)$-well-conditioned basis for $A$ if (i) $\colspan(U) = \colspan(A)$,  (ii) $\norm{U}_p \le \alpha$ and (iii) $\norm{z}_q \le \beta \norm{U z}_p$ for all $z \in \R^d$, where $q = p/(p-1)$ is the conjugate index of $p$.
\end{definition}

\begin{theorem}[\cite{DDH+09}] \label{thm:well-conditioned-basis}
  Let $A$ be an $n \times d$ matrix of rank $r$, $p \in [1, \infty)$ and $q$ be the conjugate index of $p$. There exists an $(\alpha, \beta, p)$-well-conditioned basis $U$ for the column space of $A$ such that:
\begin{enumerate}[label=(\arabic*)]
\item if $p<2$ then $\alpha = r^{\frac{1}{p} + \frac{1}{2}}$ and $\beta = 1$,
\item if $p = 2$ then $\alpha = \sqrt{d}$ and $\beta = 1$, and
\item if $p > 2$ then $\alpha = r^{\frac{1}{p} + \frac{1}{2}}$ and $\beta = r^{\frac{1}{p} - \frac{1}{2}}$.
\end{enumerate}
Moreover, there is a deterministic procedure that computes a decomposition $A = UT$, where $U\in\R^{n\times r}$ is a well-conditioned basis as described above and $T\in \R^{r\times d}$ is of full row rank, in time $O(ndr + n d^5 \log n)$ for $p \ne 2$ and $O(ndr)$ if $p = 2$.
\end{theorem}

Suppose that $A_{i - 1}$ has a decomposition $A_{i-1} = UT$ as in Theorem~\ref{thm:well-conditioned-basis} and assume that $a_i \in \rowspan(A_{i-1}) = \rowspan(T)$. Since $T$ has full row rank, it has a right inverse $T^\dagger\in \R^{d\times r}$ such that $TT^\dagger = I$ and $\colspan(T^\dagger) = \rowspan(T)$. Now, the online $\ell_p$ sensitivity of $a_i$ becomes
\[
\sup_{x\in \R^r\setminus\{0\}} \frac{|\inner{a_i, T^\dagger x}|^p}{\norm{A_{i - 1} T^\dagger x}_p^p} =
\sup_{x\in \R^r\setminus\{0\}} \frac{|\inner{b, x}|^p}{\norm{Ux}_p^p} = \sup_{x \in \bS^{r-1}} \frac{|\inner{b, x}|^p}{\norm{Ux}_p^p}\;,
\]
where $b = (T^\dagger)^{\top} a_i$. The definition of the well-conditioned basis indicates that $1/\poly(d) \le \norm{Ux}_p^p \le \poly(d)$. Suppose that $x_0$ is the vector that attains the supremum. Then $|\inner{b, x_0}|^p \le |x_0|_2^p \cdot \norm{b}_2^p = \norm{b}_2^p$, implying that  $s_i^\OL(A) \le \poly(r) \cdot \norm{b}_2^p$. On the other hand, taking $x = b / \norm{b}_2$ leads to $s_i^\OL(A) \ge \norm{b}_2^p / \poly(r)$. 
Therefore, $\norm{b}_2^p$ is a $\poly(d)$-approximation to $s_i^\OL(A)$. Putting everything together, we obtain the following theorem. 

\begin{theorem}
\label{thm:streaming}
Let $A = a_1 \circ \cdots \circ a_n$ be a stream of $n$ rows, where $a_i \in \R^{1 \times d}$. There is an algorithm which computes a subset $B$ of $m = \tilde{O}(\eps^{-\frac{2(d - 1)}{d + 2p}})$ rows of $A$ and an associated weight vector $w\in \R^m$ such that with high probability for every $x \in \bS^{d - 1}$,
\[
\left|\sum_{i} w_i |\inner{B_i, x}|^p - \norm{Ax}_p^p \right| = O(\eps) \cdot  \norm{Ax}_p^p \; .
\] 
Moreover, the algorithm can be implemented in $\tilde{O}(\eps^{-\frac{2(d - 1)}{d + 2p}} + \log^{\frac{3d + 2p - 2}{d + 2p}} n)$ words of space. 
\end{theorem}

\subsection{Faster Update Time}

In this section, we give a different one-pass streaming algorithm which achieves $O(\eps)$-additive error with $O(1)$ update time, assuming that $\norm{A_i}_2 = O(1)$. The basic idea is to partition the sphere, rather than the points, into a number of regions and maintain a sketch for each region individually. 
This idea of sphere partitioning was previously used by Bourgain and Lindenstrauss~\cite{BL88} to obtain suboptimal results for $\ell_1$-subspace embeddings for $d\geq 5$.
We remark that our new algorithm does not give a subspace embedding, unlike the previous ones. 

To begin with, we state a partition lemma for spheres.

\begin{lemma}
\label{lem:partition_sphere}
Suppose that $\eta\in (0,1/2)$. There exists a partition of $\bS^{d-1}$ with $c_1(1/\eta)^{d - 1}$ regions such that (i) the diameter of each region is at most $2\eta$ and (ii) for every $x \in \bS^{d - 1}$, the hyperplane $\{y: \inner{x, y} = 0\}$ only intersects at most $c_2(1/\eta)^{d - 2}$ regions. Here $c_1$ and $c_2$ are constants that only depend on $d$.%
\end{lemma}

\begin{proof}
Take a maximal set $\cN \subset \bS^{d-1}$ such that $d(x,y) > \eta$ for all distinct $x, y\in \cN$. It is a standard fact that $m := \abs{\cN}\leq c_1(d)/\eta^{d - 1}$. Suppose that $\cN = \{v_1,\dots,v_m\}$. For each $i$, define $Q_i = \{x\in \bS^{d-1}: d(x,v_i) = d(x,\cN)\}$ and $R_i = Q_i\setminus \bigcup_{j<i} R_j$. We can see the regions $R_1,\dots,R_m$ form a partition of $\bS^{d-1}$. By the construction of $\cN$, we can see that $\interior(B(v_i,\eta/2))\subseteq R_i\subseteq B(v_i,\eta)$, where $\interior(\cdot)$ denotes the relative interior of a set on $\bS^{d-1}$.

Now, given any $x \in \bS^{d - 1}$, we bound the number of the regions that intersect the equator $E_x = \{ y\in \bS^{d-1}: \langle x,y\rangle = 0\}$. On the one hand, if a region $R_i$ intersects the equator $E_x$, it is covered by the band $\{y\in \bS^{d-1}: d(y,E_x)\leq 2\eta\}$, which has area at most $c_2(d) \eta$. On the other hand, each region $R_i$ contains an open neighbourhood $B(v_i,\eta/2)$, thus having area at least $c_3(d)\eta^{d-1}$. Since the regions are disjoint, the number of regions that intersect $E_x$ must be at most $c_2(d)\eta / c_3(d)/\eta^{d-1} = c_4(d) / \eta^{d-2}$.
\end{proof}

\paragraph{Algorithm Description.} We now describe our algorithm. Given approximation parameter $\eps$, invoking Lemma~\ref{lem:partition_sphere} with $\eta = \eps^{\frac{2}{d + 2p - 1}}$, we obtain a partition of the sphere $R_1, \dots, R_t$ where $t = O(\eta^{d - 1}) = O(\eps^{-\frac{2(d - 1)}{d + 2p - 1}})$. 
For each region we maintain the following two things throughout the stream (we rescale the points which are not on the sphere to unit vectors when determining the region for them): (i) the sum of the $p$-th tensor products over the points in the region $q_i = \sum_{y \in R_i} y^{\otimes p}$. (ii) a sample point $z_i$ in the region, for which we employ reservoir sampling so that the sample point is uniformly chosen from all the points in the region. If the number of points in one region exceeds $\eta^{d - 1} n$ in the middle of the stream, we will treat this region as a new region and create a separate sketch for the new incoming points in it. Note that such an operation will create at most $O((1/\eta)^{d - 1})$ new regions.

\paragraph{Query Algorithm.} Given a query point $y \in \bS^{d - 1}$, we perform the following procedure: for each region $R_i$, similar to the analysis in Section~\ref{sec:upper_bound}, if the hyperplane $H = \{y: \inner{x, y} = 0\}$ does not intersect $R_i$, we have that 
\[
\inner{x^{\otimes p}, \sum_{y \in R_i}  y^{\otimes p}} = \sum_{y \in R_i} \inner{x, y}^p = \inner{x^{\otimes p}, q_i} \;.
\]
Hence we obtain $\sum_{y \in R_i} |\inner{x, y}|^p$ with zero error. Now we focus on the case that $H$ intersects $R_i$. In this case, since the value of $|\inner{x, y}|^p$ is small for all $y \in R_i$, we can use the sample point $z_i$ to estimate $\sum_{y \in R_i} |\inner{x, y}|^p$. Specifically, suppose that $R_i$ contains $c_i$ points. We define an estimator $Z_i = |\inner{x, z}|^p \cdot c_i$ and let $Z = \sum_{i \in I} Z_i$ be our final estimator, where $I$ is the set of indices of the regions which $H$ intersects.

It is easy to see that $Z$ is an unbiased estimator. To analyze the concentration, note that for each such area $R_i$, $\Var[Z_i] \le \sum_{v \in R_i} \frac{1}{c_i}\cdot(c_i \cdot \ell)^2 = c_i^2 \ell^2$, where $\ell = \max_{v \in R_i} |\inner{v, x}|^p \leq \eta^p$. Recall that $c_i \le \eta^{d - 1} n$ by our construction, and there are $O(\eta^{-(d - 1)})$ regions. We have that $\Var[Z_i] \le \eta^{2(d - 1)} n^2 \eta^{2p}$ and thus $\Var[Z] \le \eta^{d - 1 + 2p} n^2 = \eps^2 n^2$. It then follows from Chebyshev's inequality that the error $|Z - \E Z| = O(\eps n)$ with probability at least $9/10$. Rescaling by a normalization factor of $1/n$  yields the following theorem. 

\begin{theorem}[For-each version] \label{thm:for-each-general-d}
Let $A = a_1 \circ \cdots \circ a_n$ be a stream of $n$ rows, where $a_i \in \R^{1 \times d}$. There is an algorithm which maintains a data structure $Q$ of $O(\eps^{-\frac{2(d - 1)}{d + 2p - 1}})$ words of space such that for each $x \in \bS^{d - 1}$, with probability at least $9/10$, 
\begin{equation}\label{eqn:data-structure-guarantee}
\left|Q(x) - \frac{1}{n}\norm{Ax}_p^p \right| = O(\eps)  \; .
\end{equation}
Moreover, the algorithm updates $Q$ in $O(1)$ time and can be implemented in $O(\eps^{-\frac{2(d - 1)}{d + 2p -1}})$ words of space. 
\end{theorem}

Applying the median trick and a net argument, we obtain the following for-all version of Theorem~\ref{thm:for-each-general-d}.

\begin{corollary}[For-all version]
Let $A = a_1 \circ \cdots \circ a_n$ be a stream of $n$ rows, where $a_i \in \R^{1 \times d}$. There is an algorithm which maintains a data structure $Q$ of $\tilde{O}(\eps^{-\frac{2(d - 1)}{d + 2p - 1}})$ words of space such that \eqref{eqn:data-structure-guarantee} holds for all $x\in \bS^{d-1}$ simultaneously with high probability. Moreover, the algorithm updates $Q$ in $O(1)$ time and  can be implemented in $\tilde{O}(\eps^{-\frac{2(d - 1)}{d + 2p - 1}})$ words of space.
\end{corollary}

\paragraph{Tight Bound for $d \le 2p + 2$.} Below we show that when $d \le 2p + 2$, the algorithm above can achieve a tight space complexity with a small modification. Let $\eta = \eps^{\frac{2}{d + 2p}}$. When a region contains more than $\eta^{d - 1} n$ points, we split it into a number of new sub-regions with diameter at least half that of the larger region and we stop splitting when a region has diameter less than $O(\poly(\eps))$ because the error is negligible at this point. Specifically, for a region with diameter $O(\eta / 2^i)$, we take a new $O(\eta / 2^{i + 1})$-partition in Lemma~\ref{lem:partition_sphere} and partition the region using the new partitions. 

Now we turn to bound the variance. We split $Z$ as $Z = \sum_{i\geq 0} \sum_{j\in I_i} Z_j$, where $I_i$ is the set of indices of regions which $H$ intersects and has diameter in $(\eta/2^{i-1}, \eta/2^i]$. 
From Lemma~\ref{lem:partition_sphere} we have that $|I_i| = O(\eta^{-(d - 2)}\cdot 2^{i(d - 2)})$ and so $\sum_{j\in I_i} \Var[Z_j] \leq (\eta^{d - 1}n)^2 \cdot (\eta / 2^{i})^{2p} \cdot |I_i| = O(\eta^{d + 2p} n) = O(\eps^2 n^2)$. Summing over $i=0,1,\dots, O(\log(1/\eps))$, it follows that $\Var[Z] = O(\eps^2 n^2 \log(1/\eps))$. 
Using the same argument as before and after a rescaling of $\eps$, we conclude with the following theorem.

\begin{theorem} \label{thm:tight_bound_small_d}
Suppose that $d \le 2p + 2$. Let $A = a_1 \circ \cdots \circ a_n$ be a stream of $n$ rows, where $a_i \in \R^{1 \times d}$. There is an algorithm which maintains a data structure $Q$ of $\tilde{O}(\eps^{-\frac{2(d - 1)}{d + 2p}})$ words of space such that
\eqref{eqn:data-structure-guarantee} holds for all $x\in \bS^{d-1}$ simultaneously with high probability. Moreover, the algorithm updates $Q$ in $O(1)$ time and can be implemented using $\tilde{O}(\eps^{-\frac{2(d - 1)}{d + 2p}})$ words of space.
\end{theorem}

\section{Affine $\ell_p$-Subspace Sketch}
\label{sec:affine_subspace_sketch}

In this section, we consider the following special version of the affine $\ell_p$-subspace sketch problem, which can be seen as a generalization of the $\ell_p$ subspace sketch problem. Given a matrix $A \in \R^{n \times d}$ and for any given $x \in \R^d$ and $b \in \R$, we want to determine a $(1 \pm \eps)$-approximation to $\sum_i |\inner{A_i, x} - b|^p$.

\paragraph{Upper Bound.} We first consider the upper bound. The crucial observation is that, given matrix $A \in \R^{n \times d}$, let $B \in \R^{n \times (d + 1)}$ be the matrix for which the $i$-th row $B_i = \begin{pmatrix}A_i & -1\end{pmatrix}$. Suppose that $x \in \R^d$ and $b \in \R$ are the query vector and value, and let $y  = \begin{pmatrix}x^\top & b\end{pmatrix}^\top \in \R^{d + 1} $. Then
\[
\sum_i |\inner{A_i, x} - b|^p = \sum_i |\inner{B_i, y}|^p = \norm{By}_p^p.
\]
Hence, any data structure that solves the $\ell_p$ subspace sketch problem in $(d + 1)$ dimensions solves the affine $\ell_p$ subspace sketch problem for $d$ dimensions. The following theorem follows immediately from Theorem~\ref{thm:coreset_multiplicative}.

\begin{theorem}
Suppose that $A \in \R^{n \times d}$ and $p$ is a positive integer constant. Then there is a polynomial-time algorithm that returns a data structure with $\tilde{O}(\eps^{-\frac{2d}{d + 2p + 1}})$ bits of space, which solves the affine $\ell_p$ subspace sketch problem. 
\end{theorem}

\paragraph{Lower Bound.} We now turn to the lower bound. We will show that interestingly, we obtain a lower bound with the same space complexity. To achieve this, we need the following lemma, which can be seen as a stronger version of Lemma~\ref{lem:min_error}.

\begin{lemma}
\label{lem:min_error_weight}
Suppose that $p\in [1,\infty)\setminus 2\Z$ is a constant.
Let $A$ and $B$ be sets of $n\leq N$ points in a spherical shell $\{x\in \R^d: \alpha \leq \norm{x}_2 \leq \beta\}$, where $\alpha,\beta>0$ are constants such that $\alpha<\beta<(\frac{1+\sqrt{3}}{2})^{\frac{1}{p}}\alpha$. Suppose that $A$ and $B$ are symmetric around the origin and $\norm{\frac{A_i}{\norm{A_i}_2} - \frac{B_j}{\norm{B_j}_2}}_2 \geq \eta = C_1 N^{-\frac{1}{d - 1}} $ for all $i\neq j$. Then we have
\[
\delta \equiv \sup_{x \in \bS^{d - 1}} \frac{1}{n}\Abs{ \norm{Ax}_p^p - \norm{Bx}_p^p } \ge c_2 N^{-\frac{d + 2p}{2(d - 1)}}.
\]
 In the statement above, $C_1 > 0$ is a constant that depends only on $d$ and $c_2 > 0$ is a constant that depends on $d,p,\alpha,\beta$ only.
\end{lemma}

\begin{proof}
Normalizing the points in $A$ and $B$, we consider the equivalent form of this problem: $A$ and $B$ are still subsets of $\bS^{d - 1}$, while the target error $\delta$ becomes
\[
\delta \equiv \sup_{x \in \bS^{d - 1}} \left( \sum_i w_{A_i}|\inner{A_i, x}|^p - \sum_i w_{B_i} |\inner{B_i, x}|^p\right),
\]
where $w_i$'s are weights such that $w_i \in [\frac{1}{\beta^p n},\frac{1}{\alpha^p n}]$. Define the function $h_A = \sum_i w_{A_i}|\inner{A_i, x}|^p$ and $h_B = \sum_i w_{B_i} |\inner{B_i, x}|^p$. Following similar steps as in the proof of Lemma~\ref{lem:min_error}, we obtain that
\[
   \sum_{i = 1}^{n} \frac{w_{A_i}(1 - r^2)}{(1 + r^2 - 2r\inner{u, A_i})^{d/2}} = c_A + \sum_{\text{even }k\geq 2} r^k \lambda_k^{-1} \sum_j \inner{h_A, Y_{k,j}} Y_{k, j}(u) \;
\]
and
\[
    \sum_{i = 1}^{n} \frac{w_{B_i}(1 - r^2)}{(1 + r^2 - 2r\inner{u, B_i})^{d/2}} = c_B + \sum_{\text{even }k\geq 2} r^k \lambda_k^{-1} \sum_j \inner{h_B, Y_{k,j}} Y_{k, j}(u) ,
\]
where $c_A = \sum_i w_{A_i}$ and $c_B = \sum_i w_{B_i} $ are both contained in $[\frac{1}{\beta},\frac{1}{\alpha}]$.
Hence
\begin{multline*}
    \sum_{i = 1}^{n} \frac{w_{A_i}(1 - r^2)}{(1 + r^2 - 2r\inner{u, A_i})^{d/2}} - \sum_{i = 1}^{n} \frac{w_{B_i}(1 - r^2)}{(1 + r^2 - 2r\inner{u, B_i})^{d/2}} \\ 
    = c_A - c_B + \sum_{\text{even }k\geq 2}  r^k \lambda_k^{-1} \sum_j \inner{h_A - h_B, Y_{k,j}} Y_{k, j}(u) .
\end{multline*}
Integrating with respect to $u$ on $\bS^{d-1}$, we have that
\begin{equation}\label{eq:7}
\begin{aligned}
    &\quad\ \left\|\sum_{i = 1}^{n} \frac{w_{A_i}(1 - r^2)}{(1 + r^2 - 2r\inner{u, A_i})^{d/2}} - \sum_{i = 1}^{n} \frac{w_{B_i}(1 - r^2)}{(1 + r^2 - 2r\inner{u, B_i})^{d/2}}\right\|_{L_2(\sigma_{d-1})}^2 \\
    &= \sum_{\text{even }k\geq 2} r^{2k} \lambda_k^{-2} \sum_j \inner{h_A - h_B, Y_{k,j}}^2 + (c_A - c_B)^2\\
    &\le \delta^2 \max_{\text{even }k\geq 2}(r^{2k} \lambda_k^{-2}) + (c_A - c_B)^2 ,
\end{aligned}
\end{equation}
Following the same steps as in Lemma~\ref{lem:min_error}, we obtain that the left-hand side of~\eqref{eq:7} is at least $\frac{2}{\beta^{2p}}-\frac{1}{\alpha^{2p}} > 0$. Note that we have $c_A$ and $c_B$ are both in $[\frac{1}{\beta^p},\frac{1}{\alpha^p}]$ so $(c_A - c_B)^2 \geq (\frac{1}{\alpha^p}-\frac{1}{\beta^p})^2$. Hence 
\[
\delta^2 \max_{\text{even }k}(r^{2k} \lambda_k^{-2}) \ge \frac{2}{\beta^{2p}}-\frac{1}{\alpha^{2p}} - (\frac{1}{\alpha^p}-\frac{1}{\beta^p})^2 = \frac{1}{\beta^{2p}} + \frac{2}{\alpha^p\beta^p} - \frac{2}{\alpha^{2p}}  > 0
\]
by our assumptions on $\alpha$ and $\beta$. It follows that
\[
\delta \ge c(d,p,\alpha,\beta) N^{-\frac{d + 2p}{2(d - 1)}}  . \qedhere
\]
\end{proof}
Consider the same $S = \{p_1, \ldots, p_n\}$ as in Section~\ref{sec:lower_bound} for $(d + 1)$ dimensions. Consider the spherical cap
\[
C = \left\{x \in \bS^{d}: x_{d + 1} \ge \left(\frac{3}{4}\right)^{\frac{1}{p}}\right\} . 
\]
and let $T = S \cap C$. From Lemma~\ref{lem:spherical cap} we have that the area of $C$ is at least $c(d,p)$, a constant depending only on $d$ and $p$. Hence, $|T| = c(d,p) \cdot \Omega(\eps^{-\frac{2d}{d + 2p + 1}})$. Let $T'\subset \R^{d+1}$ be the set of points obtained by scaling the $(d + 1)$-st coordinate of each point in $T$ to $1$. Note that every point in $T'$ has the same last coordinate and $\ell_2$ norm in $[1,(\frac{4}{3})^{1/p}]$. Following the same steps in Section~\ref{sec:lower_bound} and combining with Lemma~\ref{lem:min_error_weight}, we have that any affine $\ell_p$ subspace sketch data structure solves the same subset identification problem. Our theorem follows immediately.

\begin{theorem}
\label{thm:lower_bound_affine}
Suppose that $p\in [1,\infty)\setminus 2\Z$ and $d$ are constants. Any data structure that solves the affine $\ell_p$-subspace sketch problem for dimension $d$ and accurary parameter $\eps$ requires $\Omega(\eps^{-\frac{2d}{d + 2p + 1}})$ bits of space.
\end{theorem}

\section{Point Estimation for SVMs} 
As an application of our results, we obtain tight space bounds for point estimation for the streaming SVM problem. In this section, we consider the following (regularized) SVM objective function:
\begin{equation}\label{eqn:F_lambda}
  F_\lambda(\theta, b) := 
\frac{\lambda}{2}\|(\theta, b)\|_2^2 + \frac{1}{n}\sum_{i=1}^n \max\{0,1-y_i(\theta^\top x_i + b)\},
\end{equation}
where $n$ data points $(x_i,y_i)\in
\R^d\times \{-1, +1\}$, with $\normtwo{x_i} = O(1)$ and $(\theta, b)\in
 \R^{d} \times \R $ are the unknown model parameters, and $\lambda$ is the regularization parameter. 
In this section, we consider the point estimation problem, that is, given $(\theta, b)$, we want to output an approximation to $F_{\lambda}(\theta, b)$. 
As mentioned in~\cite{ABL+20}, when $d \ge 2$, it is impossible to obtain a $(1 \pm \eps)$-approximation in $o(n)$ space and so we consider  $O(\eps)$-additive error.
Throughout this section, we assume that $\lambda = O(1)$, which is a common setting for this problem.

\begin{definition}[SVM Point Estimation]
    Given $(\theta, b)$ with $\norm{(\theta, b)}_2 = O(1)$, compute a value $Z$ such that $\abs{Z - F_\lambda(\theta, b)}\leq \eps$, where $F_\lambda(\theta,b)$ is as defined in~\eqref{eqn:F_lambda} with $\norm{x_i}_2 = O(1)$ for all $i$, $\lambda = O(1)$ and $y_i\in \{-1,+1\}$ for all $i$.
\end{definition}

First, we demonstrate that it suffices to consider a simplified SVM objective, as mentioned in~\cite{ABL+20}. 
We can assume that $\lambda = 0$ because we can compute the regularization term exactly. Next, recall that $y_i = \pm 1$, and so we can estimate the contribution from the positive labels and negative labels separately and so we can assume, without loss of generality, that $y_i = 1$ for all $i$. With a further replacement of $b$ with $1 - b$, the objective is changed to
\[
F(\theta, b) := \frac{1}{n}\sum_{i=1}^n \max\{0, b-\theta^\top x_i\} \;.
\]

\paragraph{Upper Bound.} The observation is similar to that in Section~\ref{sec:affine_subspace_sketch}. Given points $x_i \in \R^d$ with positive labels, let $y_i = \begin{pmatrix}-x_i^\top & 1\end{pmatrix}^\top \in \R^{d + 1}$. Suppose that $\theta \in \R^d$ and $b \in \R$ are the query vector and value, and let $\theta'  = \begin{pmatrix}\theta^\top & b\end{pmatrix}^\top \in \R^{d + 1} $. Then
\[
F(\theta, b) = \frac{1}{n}\sum_i \max\{0, b - \theta^\top x_i\} = \frac{1}{n}\sum_i \max\{0, \theta'^\top y_i\}.
\]
It is now clear that the problem can be seen as a variant of the $\ell_1$-subspace sketch problem in dimension $d + 1$, where $|\theta^\top x_i|$ is replaced with $\max\{0, \theta^\top x_i\}$. Our earlier algorithms also work for the $\max\{x, 0\}$-loss for the subspace sketch problem. Recall the proof of Lemma~\ref{lem:coreset}: we partition the data points into groups $P_1,\dots,P_t$. Given a query point $\theta$, for a group $P_i$ for which the hyperplane $H = \{x: \inner{x,\theta} = 0\}$ does not intersect, we can still obtain the exact value of   
\[
\sum_{x \in P_i} \max\{0, \inner{\theta, x}\} =  \sum_{x \in P_i} \inner{\theta, x} \text{  or  } \sum_{x \in P_i} \max\{0, \inner{\theta, x}\} = 0 \;.
\]
For the groups which $H$ intersects, since $ \max\{0, \inner{\theta, x}\} \le |\inner{\theta, x}|$, our analysis of the concentration for the sampling method still holds. Therefore, an analogous version of Lemma~\ref{lem:coreset_space} holds, that is, we can find a coreset of size $\tilde{O}(\eps^{-\frac{2d}{d+3}})$, which approximates $F(\theta, b)$ up to an additive error of $\eps$ in $O(n)$ space. To reduce the space usage, the key observation is that uniformly sampling $O(1/\eps^2)$ points, by Hoeffding's inequality, is sufficient for an $O(\eps)$ additive error. Therefore, we have effectively reduced $n$ to $O(1/\eps^2)$ as it suffices to find a coreset for $O(1/\eps^2)$ uniformly sampled points, for which we use Algorithm~\ref{alg:1}. Thus, we arrive at the following theorem. 

\begin{theorem}
\label{thm:streaming_svm}
Let $x_1, \dots, x_n$ be a stream of $n$ points, where $x_i \in \R^{d}$. There is an algorithm which computes a subset $y_i$'s of $m = \tilde{O}(\eps^{-\frac{2d }{d + 3}})$ points of $a_i$'s and weights $w\in \R^m$ such that with high probability, for every $(\theta, b) \in \R^d \times \R$ with $\norm{(\theta, b)}_2 = O(1)$,
\[
\left|\frac{1}{n}\sum_i w_i \cdot \max\{0, b - \theta^\top y_i\} - \frac{1}{n} \sum_i  \max\{0, b - \theta^\top x_i\}\right| = O(\eps) \; .
\] 
Moreover, the algorithm can be implemented in $\tilde{O}(\eps^{-\frac{2d}{d + 3}})$ words of space. \end{theorem}

\paragraph{Lower Bound.} We now turn to lower bounds for the point estimation problem. Suppose that $X = \{x_i\}$ is the point set given by the data stream. Let $-X = \{-x: x\in X\}$ and observe that
\[
\underset{X}{F(\theta, 0)} + \underset{-X}{F(\theta, 0)} = \frac{1}{n}\sum_i \left(\max\{0, \theta^\top x_i\} + \max\{0, -\theta^\top x_i\}\right) = \frac{1}{n} \sum_i \left|\theta^\top x_i\right|.
\]
Thus, if we can solve the $d$-dimensional point estimation problem for SVM, we can solve the $d$-dimensional $\ell_1$-subspace sketch problem, whence a lower bound of $\Omega(\eps^{-\frac{2(d - 1)}{d + 2}})$ bits follows from Theorem~\ref{thm:lower_bound}. To obtain a tight bound, consider the affine $\ell_1$-subspace sketch problem in Section~\ref{sec:affine_subspace_sketch}. Specifically, we have
\[
\underset{X}{F(\theta, b)} + \underset{-X}{F(\theta, -b)} = \frac{1}{n}\sum_i \left(\max\{0, b - \theta^\top x_i\} + \max\{0, -b +\theta^\top x_i\}\right) = \frac{1}{n} \sum_i \left|b - \theta^\top x_i\right| \;,
\]
which implies that if we can solve the $d$-dimensional point estimation problem for SVM, we would be able to solve the $d$-dimensional affine $\ell_1$-subspace sketch problem. Our theorem follows immediately from Theorem~\ref{thm:lower_bound_affine}.

\begin{theorem}
\label{thm:lower_bound_svm}
Suppose that $d$ is constant. Any data structure that solves the $d$-dimensional point estimation problem for SVM requires $\Omega(\eps^{-\frac{2d}{d + 3}})$ bits of space.
\end{theorem}

We remark that our analysis shows tight space complexity via a black box reduction to the $\ell_1$ subspace sketch problem, which is much simpler than the analysis in previous work~\cite{ABL+20}.

\bibliography{reference}
\bibliographystyle{alpha}

\appendix
\section{$\ell_p$-Subspace Sketch Upper Bound for Non-integer $p$}
\label{sec:non-integer}

In this section, we consider the $\ell_p$-subspace sketch problem for non-integers $p$. We shall show that for every constant $d$ and non-integer constant $p$, it is still possible to obtain a sketch of size $o(\eps^{-2})$.

\subsection{Subspace Embedding for $d\geq 5$}\label{sec:d>=5}

In Section~\ref{sec:upper_bound}, we adapted a proof of Matousek~\cite{Mat96}, giving a bound that contains logarithmic factors. Matousek's work contains a second result, which shows a tight $O(\eps^{-2(d-1)/(d+2)})$ bound without logarithmic factors for $d\geq 5$ and $p=1$. In fact, the same bound also holds for $p>1$, which we shall demonstrate below.

The starting point and the main change is the following generalization of a proposition in~\cite{Mat96} to $p > 1$.
\begin{lemma}[Generalization of {\cite[Proposition 9]{Mat96}}]
Let $d\geq 3$ and $P\subseteq \bS^{d-1}$ be a point set of size $N$, where $N\geq N_0$ for some large constant $N_0$. There exist a subset $P^\ast\subseteq P$ of $N^\ast \geq N/8$ points with $N^\ast$ even, and a subset $Q\subseteq P^\ast$ of size $N^\ast/2$ such that for all $x\in \bS^{d-1}$
\[
\abs{ \sum_{y\in P^\ast} \abs{\inner{x,y}}^p - 2 \sum_{y\in P^\ast} \abs{\inner{x,y}}^p } = O(N^{\frac{1}{2} - \frac{3}{2(d-1)}})
\]
\end{lemma}
\begin{proof}
We only highlight the changes from the original proof in~\cite{Mat96}. Take $k$ such that $2^k\sim N^2$ and for each $i=1,\dots, k$, let $\mathcal{N}_i$ be a $2^{-i}$-net on $\bS^{d-1}$ and $\pi_i(x):\bS^{d-1}\to \mathcal{N}_i$ be the induced projection map. Then define $\phi_{i,q}\in\bS^{d-1}\to \R$ for $q \in \mathcal{N}_i$ as
\[
\phi_{i,q}(y) = \begin{cases}
				\abs{\inner{q,y}}^p, & \text{if } i =1;\\
				\abs{\inner{q,y}}^p - \abs{\inner{\pi_{i-1}(q),y}}^p, & \text{if } i > 1.
			   \end{cases}
\]
The claim is that the functions $\phi_{i,q}$ satisfy the following three properties:
\begin{enumerate}[label=(\roman*)]
	\item $|\phi_{i,q}(y)| = O(2^{-i})$ for all $y\in \bS^{d-1}$;
	\item Define 
	\begin{gather*}
		L_{i,q}^+ = \{y\in \bS^{d-1}: \inner{q,y}\geq 0 \text{ and }\inner{\pi_{i-1}(q),y}\geq 0\}\\
		L_{i,q}^- = \{y\in \bS^{d-1}: \inner{q,y}\leq 0 \text{ and }\inner{\pi_{i-1}(q),y}\leq 0\}
	\end{gather*}
	Then $\phi_{i,q}$ is $O(2^{-i})$-Lipschitz on $L_{i,q}^+$ and $L_{i,q}^-$.
	\item The expansion
		\[
			|\langle x,y\rangle|^p = \sum_{i=1}^k \phi_{i,q_i}(y) + r_x(y),
		\]
		where $q_k=\pi_k(x)$ and $q_{i-1}=\pi_{i-1}(q_i)$, has the remainder term $|r_x(y)| = O(N^{-2})$ for all $y\in \bS^{d-1}$.
\end{enumerate}
The three properties are easy to verify when $p=1$. Now we shall verify them for a general $p>1$. 
For notational convenience, let $u = \pi_{i-1}(q) - q$. Then $\norm{u}_2 \leq 2^{-i}$.
Property (i) is easy to verify as $\abs{\phi_{1,q}(y)} \leq 1$ and 
\[
\abs{\phi_{i,q}(y)} = \abs{\abs{\inner{q,y}}^p - \abs{\inner{q,y} + \inner{u,y}}^p} \leq p \abs{\inner{u,y}} \leq p 2^{-i},
\]
where we used the fact that $1-x^p\leq p(1-x)$ when $x\in(0,1)$ and thus $\abs{|a|^p-|b|^p}\leq p\abs{a-b}$ when $\abs{a},\abs{b}\leq 1$. Property (iii) is also easy to verify, as 
\[
\abs{r_x(y)} = \abs{\abs{\inner{x,y}}^p - \abs{\inner{\pi_k(x), y}}^p} \leq p \abs{\inner{x-\pi_k(x),y}} \leq p \norm{x - \pi_k(x)}_2 \leq p 2^{-k} = O(N^{-2}).
\]
Next, we verify Property (ii). Suppose that $y,z\in L_{i,q}^+$. We first consider
\[
\sup_{\substack{y,z\in L_{i,q}^+\\ \inner{u,y-z}\neq 0}}\abs{\frac{\phi_{i,q}(y) - \phi_{i,q}(z)}{\inner{u,y}-\inner{u,z}}} = \sup_{\substack{y,z\in L_{i,q}^+\\ \inner{u,y-z}\neq 0}}\abs{ \frac{ \inner{q,y}^p - (\inner{q,y} + \inner{u,y})^p - \inner{q,z}^p + (\inner{q,z} + \inner{u,z})^p }{ \inner{u,y}-\inner{u,z} } }
\]
By relating the expression to the definition of derivatives, it is easy to see that this supremum is upper bounded by a constant $L$ (which depends on $p$ only), thus
\[
\abs{\phi_{i,q}(y) - \phi_{i,q}(z)} \leq L \abs{\inner{u,y-z}} \leq 2L \norm{u}_2 = O(2^{-i}).
\]
By continuity, this bound also holds for $\inner{u,y} = \inner{u,z}$. Therefore, we have verified that $\phi_{i,q}$ is $O(2^{-i})$-Lipschitz on $L_{i,q}^+$. A similar argument works for $L_{i,q}^-$ and thus we have verified (ii).

The rest of Matousek's original proof goes through, establishing the lemma.
\end{proof}

Next, we repeat Matousek's argument. Repeatedly applying the preceding lemma yields $P_1^\ast$ and $Q_1$ from $P$, $P_2^\ast$ and $Q_2$ from $P\setminus P_1^\ast$, $P_3^\ast$ and $Q_3$ from $P\setminus (P_1^\ast\cup P_2^\ast)$, and so forth. Let $Q$ be the union of these $Q_i$'s. This shows that for any set $P$ of $N$ points, one can find a subset $P'\subset P$ of size $N/2$ such that
\[
\abs{\frac{1}{N} \sum_{y\in P} \abs{\inner{x,y}}^p - \frac{2}{N}\sum_{y\in Q} \abs{\inner{x,y}}^p } = O(N^{-\frac{d+2}{2(d-1)}}),\quad \forall x\in\bS^{d-1}.
\]
An iterative application of the step above leads to the final bound of $O(\eps^{-2(d-1)/(d+2)})$, where we can repeat a point several times to accommodate different weights. Formally, we have

\begin{theorem}
\label{thm:d>=5}
Suppose that $p \geq 1$ and $d\geq 5$ are constants, $A \in \R^{n \times d}$ and $w\in \simplex{n-1}$ is the associated weight vector. There exists a polynomial time algorithm which outputs a subset $B$ of $m = O(\eps^{-\frac{2(d - 1)}{d + 2}})$ rows of $A$ and a weight vector $v\in \simplex{m-1}$ such that with high probability it holds for every $x \in \bS^{d - 1}$,
\[
\left|\sum_i v_i |\inner{B_i, x}|^p - \sum_i w_i |\inner{A_i, x}|^p \right| = O(\eps) \cdot \left(\sum_i w_i |\inner{A_i, x}|^p \right) .
\] 
\end{theorem}

\subsection{Algorithm for all $d\geq 2$ and Non-Integer $p$} \label{sec:general_tensor_UB}
First, consider the case of $1 < p < 2$. It is known that $\ell_p^n$ $(1+\eps)$-embeds into $\ell_1^N$ for some $N = C(p) n/\eps$ (see, e.g.,~\cite{FG11}), which means for every matrix $A \in \R^{n \times d}$, there is a matrix $T \in \R^{N \times n}$ such that $\norm{TAx}_1 = (1 \pm \eps)\norm{Ax}_p$ for all $x \in \R^{d}$. Thus, we can apply our existing upper bound to $TA$ under the $\ell_1$ norm, yielding an  $\tilde{O}\big(\eps^{-\frac{2(d - 1)}{d + 2}}\big)$ upper bound. We remark that the matrix $T$ constructed in~\cite{FG11} is oblivious and has independent columns, so we can write $(TA)_{\cdot, j} = \sum_{i=1}^n T_{\cdot,i} A_{i,j}$ and compute $T_{\cdot,i} A_{i,j}$ for all $j=1,\dots,d$ sequentially for each $i=1,\dots,n$. In this manner, the entire algorithm takes $O(N) = O(n/\eps)$ words of space. We remark that such an embedding-based approach does not seem amenable to the streaming setting because it would require storing the entire $A$ to compute $(TA)_{i,\cdot} = T_{i,\cdot} A$.

Now consider $p > 2$. Let $q = \lceil p / 2 \rceil$, then for every $x \in \R^d$, 
\[
\abs{\inner{A_i, x}}^p = \left(\abs{\inner{A_i, x}}^{q}\right)^{p / q} = \abs{\inner{A_i^{\otimes q}, x^{\otimes q}}}^{p / q} .
\]
Hence, the problem can be reduced to the $d^q$-dimensional $\ell_r$-subspace sketch problem with $1 < r = p / q < 2$, which leads to an  $\tilde{O}\big(\eps^{-\frac{2(d^q - 1)}{d^q + 2}}\big)$ upper bound.
Again, the algorithm uses $O(n/\eps)$ words of space and is not amenable to the streaming setting. We remark that we cannot expect a better upper bound by reducing an $\ell_p$-subspace sketch problem to an $\ell_r$-subspace sketch problem for an integer $r > 2$ because $\ell_p^2$ does not $(1+\eps)$-embed into $\ell_r^N$ when $p>r>2$~\cite{Dor76}. The following is a short proof we include for completeness. Suppose that $\ell_p^2$ does $(1+\eps)$-embed into $\ell_r^N$. Then there exist $u,v\in \ell_r^N$ such that $\frac{1}{1+\eps}(\abs{a}^p+\abs{b}^p)^{1/p} \leq \norm{au+bv}_r \leq (1+\eps)(\abs{a}^p+\abs{b}^p)^{1/p}$ for all $a,b\in\R$. This means that $(1+\eps)^r 2^{r/p}\geq \frac{1}{2}(\norm{u+v}_r^r + \norm{u-v}_r^r) \geq \norm{u}_r^r + \norm{v}_r^r \geq \frac{2}{1+\eps}$, which is a contradiction for all $\eps$ small enough.

\subsection{Streaming Algorithm for $p > d-1$} \label{sec:p>d-1}
Next, we present a streaming algorithm when $p > d - 1$, using approximation theory on the unit sphere. Given a function $h:\bS^{d-1}\to \R$, consider its series expansion in spherical harmonics
 \[
 h(x) = \sum_{k=0}^\infty \sum_{j=1}^{M(d,k)} \inner{ h, Y_{k,j} } Y_{k,j}(x)
 \]
 and the truncation of this series of order at most $K$
 \[
 (Q_K h)(x) = \sum_{k=0}^K \sum_{j=1}^{M(d,k)} \inner{ h, Y_{k,j} } Y_{k,j}(x).
 \]
 It is a well-studied problem in approximation theory on the unit sphere that (see, e.g. \cite[Theorem 2.35]{atkinson12})
 \[
 \norm{h - Q_K h}_\infty \leq (1 + \norm{Q_K}_{C\to C}) E_{K,\infty}(h),
 \]
 where $\norm{Q_K}_{C\to C}$ is the operator norm of $Q_K$ when viewed as an operator from $C(\bS^{d-1})$ to $C(\bS^{d-1})$, and $E_{n,\infty}(f)$ is the minimum error in the $\ell_\infty$ norm of approximating $f$ by polynomials of total degree at most $K$ on $\bS^{d-1}$. It was shown in~\cite{ragozin72} that $\norm{Q_K}_{C\to C} \simeq K^{d/2-1}$ when $d\geq 3$ and it is a classical result in Fourier analysis that $\norm{Q_K}_{C\to C} \simeq \ln K$ when $d=2$.

In our problem, without loss of generality, consider $h(x) = \sum_i w_i \abs{\inner{a_i,x}}^p$ with $a_i \in \bS^{d-1}$ and $\sum_i w_i = O(1)$. Let us first consider the approximation to $f(x) = \abs{\inner{a_i,x}}^p$. It follows from the result of Ragozin~\cite{ragozin71} (see also \cite[Eq. (4.49)]{atkinson12}) that
\begin{equation}\label{eqn:approx_error_infty_norm}
	E_{K,\infty}(f) \leq \frac{C_p}{n^p},
\end{equation}
where $C_p > 0$ is a constant depending only on $p$. Thus,
\[
	E_{K,\infty}(h) \leq \sum_i w_i E_{K,\infty}(f) \leq \frac{C'_p}{n^p}
\]
and
\[
 	\norm{h - Q_K h}_\infty \lesssim_p \begin{cases}
     								\ln K / K^p  & d = 2 \\
								     1 / K^{p-d/2+1} & d\geq 3.
								 \end{cases}
\]
Therefore, we can take
\[
 K \simeq_p \begin{cases}
             (1/\eps)^{1/p} \log^{1/p}(1/\eps) & d = 2\\
             (1/\eps)^{1/(p-d/2+1)} & d \geq 3
         \end{cases}
\]
such that $\norm{h - Q_K h}_\infty\leq \eps$. For our $h(x)$, we have from Funk-Hecke Theorem that
\[
 (Q_K h)(x) = \sum_{k=0}^K \lambda_k \sum_{j=1}^{M(d,k)} Y_{k,j}(x) \left(\sum_i w_i Y_{k,j}(a_i)\right).
\]
Therefore the streaming algorithm needs only to maintain 
$\sum_i w_i Y_{k,j}(a_i)$ for each $k=0,\dots,K$ and $j=0,\dots,M(d,k)$ to output $(Q_K h)(x)$. This is clearly feasible in the data stream setting, as we can calculate for each new point $a_i$ the value of $Y_{k,j}(a_i)$ for all pairs $(k,j)$ with $k\leq K$. The number of such values to maintain is
\[
O\left(\sum_{k=0}^K M(d,k) \right) = O\left(\sum_{k=0}^K k^{d-2} \right) = O(K^{d-1}).
\]

However, this approach suffers from a precision problem when $d\geq 3$. Using the explicit expression of $Y_{k,j}$ in terms of the Gegenbauer polynomials (see, e.g.~\cite[Theorem 1.5.1]{DX13}) would cause the intermediate results to be as small as $1/K^{O(Kd)}$ or as large as $K^{O(Kd)}$, thus requiring $\tilde{O}(K)$ words of space to calculate and store the value of each $Y_{k,j}(x)$. This leads to an overall space of 
\[
	\tilde{O}(K^d) = \tilde{O}(\eps^{-2d/(2p-d+2)})
\]
words when $d\geq 3$, which is $o(\eps^{-2})$ when $p > d - 1$. This precision problem, however, does not exist when $d=2$, since the spherical harmonics degenerate to exactly sines and cosines and the intermediate values can fit in a word. Thus, the streaming algorithm uses $O(K) = \tilde{O}(\eps^{-1/p})$ words of space when $d=2$, which is close to the lower bound of $\Omega(\eps^{-1/(p+1)})$ bits.

\subsection{Discussion on Subspace Embeddings} 
In this subsection, we discuss the $\ell_p$-subspace embedding, a classical problem in the local theory of Banach spaces. Given a $d$-dimensional subspace $X\subset L_p(\bS^{d-1},\sigma_{d-1})$, let $N_p(X,\eps)$ denote the minimum number $N$ such that there is a $d$-dimensional subspace $Y\subset \ell_p^N$ with $d(X, Y) \leq 1 + \eps$. The $\ell_p$-subspace embedding problem asks to find $N_p(d,\eps) = \sup_X N_p(X,\eps)$, where the supremum is taken over all $d$-dimensional subspaces $X\subset L_p(\bS^{d-1})$. 
In this work, we focus on the dependence on $\eps$, assuming that $d$ is a constant and $p$ not an even integer. Note that when $p$ is an even integer, the dependence on $\eps$ does not exist since $X$ can always be embedded into $\ell_p^N$ isometrically. To the best of our knowledge, for constant $d$, all existing results are for $p = 1$, which has the geometric meaning in relation to approximating zonoids with zonotopes. Interested readers can refer to~\cite{handbook:21} for a survey. The case of $p=1$ and constant $d$ has been nearly perfectly resolved, as
\[
\eps^{-2(d-1)/(d+2)}\lesssim_d N_1(d,\eps)\lesssim_d \begin{cases}
                            (\eps^{-2}\log(\eps^{-1}))^{(d-1)/(d+2)} & \text{if }d = 3,4\\
                            \eps^{-2(d-1)/(d+2)} & \text{if }d = 2\text{ or }d \geq 5.
                         \end{cases}   
\]
The lower bound is due to Bourgain et al.~\cite{bourgain89}, the upper bound for $d=3,4$ to Bourgain and Lindenstrauss~\cite{BL88} (their bound for $d=4$ has a worse logarithmic term but is improved by Matousek in~\cite{Mat96}) and the tight upper bound for $d\geq 5$ to Matousek~\cite{Mat96}. The lower bound technique, which is based on spherical harmonics, can be generalized to $p > 1$, yielding that
\begin{equation}\label{eqn:N_p_lower_bound}
    N_p(d,\eps)\gtrsim_{d,p} \eps^{-2(d-1)/(d+2p)}, \quad p\in [1,\infty)\setminus 2\Z
\end{equation}
in a similar (but easier) fashion to what we did in Section~\ref{sec:lower_bound}. Our upper bound in Section~\ref{sec:upper_bound} shows that this bound is tight for odd integers $p$,
\[
    N_p(d,\eps)\lesssim_{d,p} \eps^{-2(d-1)/(d+2p)}, \quad p \text{ is odd integer}.
\]
As demonstrated in Section~\ref{sec:d>=5}, the upper bound for $N_1(d,\eps)$ extends to $p>1$ with the asymptotically identical bound, that is,
\begin{equation}\label{eqn:matousek_upper_bound}
	N_p(d,\eps)\lesssim_{d} \eps^{-2(d-1)/(d+2)}, \quad d \geq 5.
\end{equation}
The independence of $p$ is due to the fact that $|x|^p$ is almost linear around $x=1$. All other techniques for the upper bound critically rely on linearity and do not generalize to non-integers $p>1$. To the best of our knowledge, for $d=2,3,4$, no upper bounds of $o(\eps^{-2})$ have been recorded in the literature. A general upper bound of $\tilde{O}(d^{\max\{1,p/2\}}/\eps^2)$ holds for all values of $d$ and $\eps$, see, e.g.~\cite[Section 15.5]{LT91}.

A special subspace of interest in the embedding problem is $\ell_2^d$, which can be viewed as a subspace of $L_p(\bS^{d-1})$ as $\{x\mapsto \beta_d^{-1/p}\langle y,x\rangle: y\in \ell_2^d\}$ with $\beta_d = \int_{\bS^{d-1}} \abs{\inner{y,u}}^p d\sigma_{d-1}(u)$, a constant independent of $y$. The lower bound~\eqref{eqn:N_p_lower_bound} is in fact a lower bound for $N_p(\ell_2^d,\eps)$,  i.e.,
\begin{equation}\label{eqn:lower_bound_uniform_distribution}
    N_p(\ell_2^d,\eps) \gtrsim_{d,p} \eps^{-2(d-1)/(d+2p)}.
\end{equation}
In this case, the subspace embedding problem asks to find $N$ points $y_1,\dots,y_N\in\bS^{d-1}$ with weights $w_1,\dots,w_N$ such that
\begin{equation}\label{eqn:uniform_distribution_approx}
\abs{ \beta_d - \sum_{i=1}^N w_i \abs{\inner{y_i,x}}^p} \leq \eps, \quad \forall x\in \bS^{d-1},
\end{equation}

To the best of our knowledge, no upper bound of $N_p(\ell_2^d,\eps)$ is known in the literature, except for $p=1$ with the upper bound of $N_1(d,\eps)$. When $d=2$, one can take the points $y_1,\dots,y_N$ to be equidistantly distributed on $\bS^1$ and all weights $w_i = 1/N$. It is not difficult to show that \eqref{eqn:uniform_distribution_approx} holds when $N \gtrsim_p \eps^{-1/(p+1)}$. Combining with the lower bound~\eqref{eqn:lower_bound_uniform_distribution}, we have obtained a tight bound \[
  N_p(\ell_2^2,\eps)\simeq_p \eps^{-1/(p+1)}, \quad p\in [1,\infty)\setminus 2\Z.
\]
For $d\geq 3$, by considering the best polynomial approximation (not necessarily a truncation of the spherical harmonic series), it suffices to find $y_1,\dots,y_N\in \bS^{d-1}$ with weights $w_1,\dots,w_N\geq 0$ such that
\[
\abs{ \int_{\bS^{d-1}} (P_K h_x)(y) d\sigma_{d-1}(y) - \sum_{i=1}^N w_i (P_K h_x)(y_i) }\leq \eps, \quad \forall x\in \bS^{d-1},
\]
where $h_x(y) = \abs{\inner{ x, y }}^p$ is a function on $\bS^{d-1}$ and $P_K h$ is the best degree-$K$ polynomial approximation to $h$. It follows from spherical design theory that there exist such $y_1,\dots,y_N\in\bS^{d-1}$ for $N = O(K^{d-1})$ such that
\[
\int_{\bS^{d-1}} f(y) d\sigma_{d-1}(y) = \frac{1}{N}\sum_{i=1}^N f(y_i) 
\]
for all polynomials $f$ of total degree at most $K$~\cite{BRV13} and thus for all $P_K h_x$. Recalling the error estimate~\eqref{eqn:approx_error_infty_norm}, we can take $K\simeq \eps^{-1/p}$ such that
\[
    \norm{h_x - P_K h_x}_\infty \leq \eps.
\]
for all $x\in\bS^{d-1}$ and thus obtain an upper bound that \[
    N_p(\ell_2^d,\eps)\lesssim \eps^{-\frac{d-1}{p}},
\]
which is $o(\eps^{-2})$ when $p > (d-1)/2$ and is better than Matousek's upper bound \eqref{eqn:matousek_upper_bound} for $d\geq 5$ when $p > d/2 + 1$. 

Using a polynomial approximation seems to have an inherent issue which prevents obtaining a good upper bound for small $p$. The best degree-$K$ approximation error $K^{-p}$ (which is asymptotically tight) does not depend on $d$ while we need a $+d$ term in the denominator of the exponent. 
Also observe that this approach, for $d=2$, gives only an $\eps^{-1/p}$ bound instead of the tight $\eps^{1/(p+1)}$, which is another intrinsic issue with polynomial approximation, as the best approximation error is $K^{-p}$ instead of $K^{-(p+1)}$. The difficulty of approximating $|t|^p$ on $[-1,1]$ is around $t = 0$ (where $t$ corresponds to $\inner{x,y}$), and the proof of the tight upper bound kind of circumvents this issue by considering the regions around the equator $\{y: \inner{x,y} = 0\}$ separately. 

An alternative route towards $\eqref{eqn:uniform_distribution_approx}$ is via geometric discrepancy theory, employed by Linhart~\cite{linhart89}. First, rewrite \eqref{eqn:uniform_distribution_approx} as (letting $w_i=1/N$ for all $i$)
\begin{equation}\label{eqn:uniform_distribution_discrepancy}
    \abs{ \int_0^1 \abs{g^{-1}(z)}^p dz - \frac{1}{N}\sum_i \abs{g^{-1}(z_i^x)}^p } \leq \eps, \quad \forall x \in \bS^{d-1},
\end{equation}
where $g:[-1,1]\to [0,1]$ is defined as $g(t) = \sigma_{d-1}(\{y\in\bS^{d-1}: y_1 \geq t\})$ (area of a spherical cap) and $z_i^x = g(\inner{y_i,x})$. Let $f(z) = \abs{g^{-1}(z)}^p$. The left-hand side of \eqref{eqn:uniform_distribution_discrepancy} equals 
\begin{equation}\label{eqn:uniform_distribution_discrepancy_form}
    \abs{\int_0^1 f'(z) \Delta_{x}(z) dz} 
\end{equation}
where $\Delta_{x}$ is the discrepancy function of the points $\{z_i^x\}$ on $[0,1]$, defined as
\[
\Delta_{x}(z) = \frac{\abs{i: z_i^x \in [0,z]}}{N}-z.
\]
This is the starting point of Koksma's inequality. The discrepancy of $\{z_i^x\}$ on $[0,1]$ translates to the discrepancy of $\{y_i\}$ on spherical caps, for which Beck obtained  asymptotically tight upper bounds (see, e.g. \cite[Theorem 24]{beck:book}):
\[
\norm{\Delta_{x}(z)}_\infty \lesssim_{d} N^{-\frac{d}{2(d-1)}}\sqrt{\log N} , \quad \norm{\Delta_{x}(z)}_2\lesssim_{d} N^{-\frac{d}{2(d-1)}}.
\]
An upper bound (Koksma-type inequality) for \eqref{eqn:uniform_distribution_discrepancy_form} then follows as
\[
    \abs{\int_0^1 f'(z) \Delta_{x}(z) dz} \leq \norm{f'}_2 \norm{\Delta_{p,x}}_2 \lesssim_{d,p} N^{-\frac{d}{2(d-1)}}
\]
and thus
\[
	N_p(\ell_2^d,\eps) \lesssim_{d,p} \eps^{-\frac{2(d-1)}{d}}.
\]
The main part has the form $O(\eps^{-2 + O(1/d)})$, independent of $p$ but is not tight when $p=1$. This independence of $p$ comes from the fact that the discrepancy function does not depend on $p$. 
When $d=3$, Linhart~\cite{linhart89} suggests to use Lambert's area-preserving map $\phi$ to reduce the problem from $\bS^2$ to $[0,1]^2$. Here $\phi:[0,1]^2\to \bS^2$ is defined as $\phi(\tau,\alpha) = (2\sqrt{\tau-\tau^2}\cos(2\pi\alpha), 2\sqrt{\tau-\tau^2}\sin(2\pi\alpha), 1-2\tau)$. For each $x\in\bS^2$, one can define an $f_x$ on $[0,1]^2$ as $f_x(\tau,\alpha) = \abs{\inner{x,\phi(\tau,\alpha)}}^p$. Then the left hand side of \eqref{eqn:uniform_distribution_approx} can be written as
\[
\abs{ \int_0^1\int_0^1 f_x(\tau,\alpha) d\tau d\alpha - \frac{1}{N}\sum_i f_x(z_i)},
\]
where $z_i = \phi^{-1}(y_i)$. By Koksma-Hlawka inequality, 
\begin{equation}\label{eqn:uniform_distribution_discrepancy_d=3}
\abs{ \int_0^1\int_0^1 f_x(\tau,\alpha) d\tau d\alpha - \frac{1}{N}\sum_i f_x(z_i)} \leq V(f_x) \norm{D(\tau,\alpha; z_1,\dots, z_N)}_\infty,
\end{equation}
where $V(\cdot)$ denotes the variation in the Hardy-Krause sense and 
\[
D(\tau,\alpha; z_1\dots,z_N) = \frac{1}{N}\sum_{i=1}^N \mathbf{1}_{\{z_i \in [0,\tau]\times[0,\alpha]\}} - \tau\alpha
\] 
is the local discrepancy function on $[0,1]^2$. It can be verified that $V(f_x)$ is uniformly bounded by a constant depending on $p$ for all $x\in\bS^2$ and there exist $z_1,\dots,z_N$ such that $\norm{D}_\infty \lesssim (\log N)/N$ (see, e.g. \cite[Theorem 4]{beck:book}), we can thus take $N \simeq \eps^{-1}\log(\eps^{-1})$ to fulfil~\eqref{eqn:uniform_distribution_approx}. This is again independent of $p$ and is worse than the $\tilde{O}(\eps^{-4/5})$ upper bound for $N_1(d,\eps)$. It looks very difficult to analyse \eqref{eqn:uniform_distribution_discrepancy_form} to give rise to a bound for $N$ incorporating both $p$ and $d$ in the denominator of the exponent.

We can extend Linhart's approach to the non-uniform case. Note that $\phi$ maps the boundary of $[0,1]^2$ to a specific half great circle $C\subset \bS^2$. Given a finite set of points on $\bS^2$, we can assume, without loss of generality, that none of them lies on $C$ (otherwise we can rotate the points). Instead of \eqref{eqn:uniform_distribution_discrepancy_d=3}, we now have from the generalized Koksma-Hlawka inequality~\cite{Goetz2002,Aistleitner2015} that
\[
\abs{ \int_{[0,1]^2} f_x(\tau,\alpha) d\mu - \frac{1}{N}\sum_i f_x(z_i)} \lesssim V(f_x) \norm{D(\tau,\alpha; z_1,\dots, z_N;\mu)}_\infty
\]
where $\mu$ is a normalized discrete measure supported on $\{z_1,\dots,z_N\}\subset (0,1)\times(0,1)$ and 
\[
D(\tau,\alpha; z_1\dots,z_N;\mu) = \frac{1}{N}\sum_{i=1}^N \mathbf{1}_{\{z_i \in [0,\tau]\times[0,\alpha]\}} - \mu([0,\tau]\times [0,\alpha])
\] 
is the local discrepancy function with respect to $\mu$. Using the discrepancy result~\cite{Aistleitner2014} for non-uniform distributions, we can obtain, similarly to the above, that
\[
N_p(3,\eps) \lesssim_p \eps^{-1}\log^{7/2}(\eps^{-1}).
\]
More generally, the same bound (with hidden constants depending on both $p$ and $d$) holds also for larger constant $d\geq 4$ and $p > d-2$, because $V(f_x)$ in this case is finite when $\phi: [0,1]^{d-1}\to \bS^d$ corresponds to the usual spherical coordinates.

In summary, we have the following motley collection of upper bounds for non-integer $p$. For $N_p(\ell_2^d, \eps)$,
\begin{align*}
N_p(\ell_2^2,\eps) &\simeq_p \eps^{-1/(p+1)} \\
N_p(\ell_2^3,\eps) &\lesssim_p \begin{cases}
						\eps^{-1}\log(\eps^{-1}) &  \text{if }1 < p < 2\\
						\eps^{-2/p} & \text{if }p \geq 2
					\end{cases}\\
N_p(\ell_2^4,\eps) &\lesssim_p \begin{cases}
						\eps^{-8/7} &  \text{if }1 < p < 21/8\\
						\eps^{-3/p} & \text{if }p \geq 21/8 
					\end{cases}\\
N_p(\ell_2^d,\eps) &\lesssim_{d,p} \begin{cases}
						\eps^{-2(d-1)/(d+2)} & \text{if }1 < p \leq d/2 + 1\\
						\eps^{-(d-1)/p} & \text{if } p > d/2 + 1
					\end{cases}, \quad d \geq 5.
\end{align*}
Here the bound for $N_p(\ell_2^4,\eps)$ and small $p$ comes from the bound for $N_p(\ell_2^5,\eps)$ and small $p$.

For general $N_p(d,\eps)$, we have for non-integer $p > 1$ that
\begin{align*}
	N_p(d,\eps) \lesssim_{d,p} \begin{cases}
						\eps^{-1}\log^{7/2}(\eps^{-1}), & p > d - 2\\
						\eps^{-8/7}, & d = 4 \text{ and } 1 < p < 2\\
						\eps^{-2(d-1)/(d+2)}, & d\geq 5 \text{ and } 1 < p \leq d - 2
					   \end{cases}
\end{align*}
For odd integers $p$, we have the near-tight bound for all $d\geq 2$
\[
\eps^{-2(d-1)/(d+2p)} \lesssim_{d,p} N_p(\ell_2^d,\eps) \leq N_p(d,\eps) \lesssim_{d,p} (\eps^{-2}\log(\eps^{-1}))^{(d-1)/(d+2p)}.
\]

\paragraph{Closing Remarks.} We conjecture that the lower bound \eqref{eqn:N_p_lower_bound} is tight for $N_p(\ell_2^d,\eps)$ for all $p\notin 2\Z$ and leave finding a matching upper bound for non-integer $p$ as an open problem. It is also an intriguing open problem to find an upper bound for $N_p(d,\eps)$ for $d=2,3,4$ and non-integer $p$ that matches $N_1(d,\eps)$. A major open problem is to find an upper bound for $N_p(d,\eps)$ that contains both $p$ and $d$ as positive terms in the denominator of the exponent for general $d$.

\end{document}